\documentclass[a4paper,UKenglish,cleveref, autoref, thm-restate]{lipics-v2021}

\usepackage{rotate}
\usepackage{multirow}
\usepackage{multicol}
\usepackage{comment}
\usepackage{bussproofs}
\usepackage{amssymb}
\usepackage{latexsym}
\usepackage{amssymb}
\usepackage{mathrsfs}
\usepackage{graphics}
\usepackage{subcaption}
\usepackage{stmaryrd} 
\usepackage{fancyhdr}
\usepackage{tikz-cd}
\usepackage{proof}

\usepackage{esvect}
\usepackage{orcidlink}

\usepackage{color}
\usepackage[title]{appendix}

\newenvironment{bprooftree}
  {\leavevmode\hbox\bgroup}
  {\DisplayProof\egroup}

\def\namedlabel#1#2{\begingroup
    #2%
    \def\@currentlabel{#2}%
    \phantomsection\label{#1}\endgroup
}
\newcommand{\BintNA}{\sf{N2Int^*}}
\newcommand{\Bint}{\sf{2Int}}
\newcommand{\BintN}{\sf{N2Int}}
\newcommand{\BintNs}{\sf{N2Int}^*}

\newcommand{\Pts}{\sf{PtS}}
\newcommand{\Bes}{\sf{BeS}}
\newcommand{\At}{\sf{At}}
\newcommand{\nn}{\sim\!\!}
\newcommand{\Nt}{\sf{N3}}
\newcommand{\Nf}{\sf{N4}}
\newcommand{\Nfb}{\Nf^{\bot}}

\newcommand{\U}{\mathcal{U}}

\newcommand\eg{\hbox{{\em e.g.}}}
\newcommand{\B}{\mathcal{B}}
\newcommand{\C}{\mathcal{C}}
\newcommand{\sS}{\sf{S}}
\newcommand{\BPR}{\sf{BPR}}

\newcommand{\Victor}[1]{\textcolor{cyan}{Victor: #1}}
\title{Bilateralism with incompatible proofs and refutations}

\author{Victor Barroso-Nascimento}{Department of Computer Science, University College London, UK}{victorluisbn@gmail.com}{https://orcid.org/0000-0002-3990-5996}{}

\author{Maria Os\'{o}rio}{Faculty of Sciences, University of Lisbon, Portugal}{mocosta@fc.ul.pt}{https://orcid.org/0009-0005-5322-1223}{}

\author{Elaine Pimentel}{Department of Computer Science, University College London, UK}{e.pimentel@ucl.ac.uk}{https://orcid.org/0000-0002-7113-0801}{}

\authorrunning{V. Barroso-Nascimento and M. Os\'{o}rio and E. Pimentel}

\Copyright{Victor Barroso-Nascimento and Maria Os\'{o}rio and Elaine Pimentel}


 %TODO mandatory; please add comma-separated list of keywords

\EventEditors{John Q. Open and Joan R. Access}
\EventNoEds{2}
\EventLongTitle{42nd Conference on Very Important Topics (CVIT 2016)}
\EventShortTitle{CVIT 2016}
\EventAcronym{CVIT}
\EventYear{2016}
\EventDate{December 24--27, 2016}
\EventLocation{Little Whinging, United Kingdom}
\EventLogo{}
\SeriesVolume{42}
\ArticleNo{23}
\begin{document}

\nolinenumbers

\maketitle

\begin{abstract}
Logical bilateralism challenges traditional concepts of logic by treating assertion and denial as independent yet opposed acts. While initially devised to justify classical logic, its constructive variants show that both acts admit intuitionistic interpretations. This paper presents a bilateral system where a formula cannot be both provable and refutable without contradiction, offering a framework for modelling 
mathematical proofs and refutations that exclude inconsistency. 
We formalise the logic via a bilateral natural deduction system with the desirable proof-theoretic properties  
of normalisation, subformula property and consistency, together with a base-extension semantics grounded in explicit proofs and refutations. Finally, refutation is shown to coincide with Nelson's constructive falsity, extending intuitionistic logic for constructive epistemic reasoning.
%
%The logic is formalised through a bilateral natural deduction system with desirable proof-theoretic properties, including normalisation. We also introduce a base-extension semantics requiring explicit constructions of proofs and refutations while preventing them from being simultaneously established for the same formula. The semantics is proven sound and complete with respect to the calculus. Finally, we show that our notion of refutation corresponds to David Nelson's constructive falsity, extending rather than revising intuitionistic logic and reinforcing the system's suitability for representing constructive epistemic reasoning.
\end{abstract}

%\newpage

\section{Introduction}
A central issue in constructive logic concerns the nature of refutation and how falsity should be represented within a constructive framework. In this context, David Nelson~\cite{NelsonConstructiveFalsity} was the first to highlight a subtle but important distinction between two kinds of refutative results that can be established constructively for statements of the form $\forall x \phi$. 

\noindent
{\bf About falsity.} The first, %which we may call 
{\em weak falsity}, is demonstrated by exhibiting a procedure that transforms any construction of the mathematical object into a construction of a contradiction. This is the notion of falsity traditionally adopted in constructive logics, as witnessed by the fact that $\neg \phi$ is often defined as $\phi \to \bot$. In contrast, {\em strong falsity} (represented by Nelson's strong negation $\nn\phi$) is proven by explicitly constructing witnesses to the falsity of the statement itself.

For example, let $\varphi=$ ``there exists an odd perfect number'' and define $\phi(0) = \varphi, \phi(1)=\neg \varphi$.
To establish the weak falsity of a statement such as $\forall n\in\{0,1\} \phi(n)$, one assumes it and derives a contradiction. On the other hand, to establish its strong falsity, one must instead {\em construct} one object $n$ for which $\phi(n)$ does not hold.
The significance of this difference resides in the fact that, while from the existence of a counterexample to a claim we can trivially obtain a procedure reducing purported constructions of its object to a proof of absurdity, it is not possible in general to obtain a particular counterexample to a claim from a procedure reducing it to a contradiction. 
%
%After proving that $P(a)$ does not hold for a particular $a$ and assuming that $P(x)$ holds for all $x$ we may obtain a contradiction by simply substituting $x$ by $a$. However, from the existence of a proof that assuming $P(x)$ for all $x$ leads to a contradiction, it does not follow that we can exhibit a specific object $a$ for which $P(a)$ fails.  
%
This is exactly why $\neg \forall x \phi(x)$ does not imply $\exists x \neg \phi (x)$ in, \eg, intuitionistic or minimal logic.

Nelson argues in~\cite{NelsonConstructiveFalsity} that this novel conception of strong negation naturally entails a revision of the principles of intuitionistic logic, thereby explaining why these systems are often regarded as {\em alternatives} to intuitionism rather than mere extensions of it. 
In fact, although both weak and strong falsity may reasonably be regarded as constructive, the contents of a proof of strong negation are substantially stronger. 

This also leads to a surprising (constructive) recovery of some classical equalities. In Nelson's original definitions, from a proof of $\nn\forall x \phi(x)$ it follows that the existence of a particular $a$ for which $\nn \phi(a)$ holds, hence also $\exists x\!\nn \phi(x)$. Conversely, since a proof of $\exists x\!\nn \phi(x)$ requires production of some object $a$ for which $\nn \phi(a)$ holds, if such a proof is available $a$ itself is an effective counterexample to the claim that $\forall x \phi(x)$, hence it follows that $\nn\forall x \phi(x)$. This means that, just as in classical logic, $\exists x\!\nn \phi(x)$ is equivalent to $\nn\forall x \phi(x)$, but {\em for entirely different reasons}! In classical logic the equivalence arises because the conditions for proving an existential statement are {\em weaker}, whereas in Nelson's logic it arises because the conditions for proving a negation are {\em stronger}.

\noindent
{\bf Falsity and bilateralism.} Constructively proving a formula $\phi$ requires providing a {\em constructive proof} of $\phi$. Dually, proving the negation $\nn\phi$ requires a {\em constructive disproof}--a constructive method that demonstrates that $\phi$ cannot hold.
Hence Nelson's constructive falsity can be interpreted as putting the contents of positive and negative statements on equal footing. This is essentially what is also done in {\em logical bilateralism}, an approach to logic in which positive and negative statements are not taken as mutually interdefinable. Bilateralism was originally presented by Rumfitt~\cite{Rumfitt2000-RUMYAN-2} to provide a proof-theoretic justification of classical logic, but it was later repurposed for intuitionism and its falsificationalist dual \cite{ShramkoYaroslav} through Wansing's logic $\Bint$~\cite{Wan110.1093/logcom/ext035}.
%, presented by Wansing in \cite{Wan110.1093/logcom/ext035}. 
%
In short, $\Bint$ is a bilateral bi-intuitionistic logic~\cite{rauszer1980algebraic,DBLP:conf/tableaux/Gore00}, combining the verificationist account of constructivism expressed through intuitionistic logic and its falsificationalist dual-intuitionistic counterpart in a single, unified framework. %$\Bint$ is shown not only to be a conservative extension of both intuitionistic and dual-intuitionistic logic, but also to be embeddable in both. 

\noindent
{\bf Stronger strong falsity.}
One can associate the {\em disprovability} of $\phi$ with the constructive strong negation $\nn \phi$: $\phi$ is disprovable exactly when $\nn \phi$ is provable, and  $\nn \phi$ is disprovable exactly when $\phi$ is provable.
The main formalisms of this idea, proposed by Nelson, are the logical systems $\Nt$ and $\Nf$~\cite{NelsonConstructiveFalsity,Odintsov2008-ODICNA}. The system $\Nt$, essentially intuitionistic logic extended with strong negation, was designed to provide a formal account of falsity in mathematics. By contrast, $\Nf$ is obtained from $\Nt$ by removing the {\em explosion principle} $\phi \to (\nn\phi \to \psi)$, which allows arbitrary conclusions to be derived from a contradiction. As a result, $\Nf$ is {\em paraconsistent}, in the sense that contradictions do not automatically trivialise the system.

Semantically, this distinction has a clear consequence: $\Nf$ allows for {\em simultaneous} proofs and refutations of the same formulas, whereas $\Nt$ does not~\cite[p.~447]{Wan110.1093/logcom/ext035}\footnote{The  model theory of both $\Nt$ and $\Nf$ have also been extensively studied (see, \eg, \cite{Hasuo2003-HASKCO,DBLP:journals/jancl/JarvinenRR24}). A concise overview of Nelson's logics is provided in \cite{Nascimento2020-NASNL-2}, which also introduces a semantics for Nelson's lesser-known but conceptually rich system $\sS$, presented as a possible alternative to both $\Nt$ and intuitionistic logic~\cite{DavidNelsonS}.}.
This contrast suggests that $\Nf$ and $\Nt$ capture {\em distinct notions} of strong falsity. For instance, if one aims to model the set of {\em beliefs} of an agent, it is reasonable to allow both truth and falsity to coexist without collapse, since agents may hold inconsistent beliefs without accepting everything. Wansing's $\Bint$ captures well this $\Nf$ perspective. By contrast, when modelling falsity in {\em mathematics}, such coexistence is problematic: a proof of a statement cannot coherently stand alongside its refutation. The present work is devoted to this $\Nt$ view of strong falsity\footnote{For additional sources containing in-depth defences of $\Nt$ and strong negation from a conceptual perspective, the reader is referred to \cite{Kapsner2014-KAPLAF,KurbisProofandFalsity2019-KRBPAF}.}.

\noindent
{\bf Contributions and challenges.} 
We investigate the proof-theoretic structure of falsity by introducing $\BPR$, a natural deduction system for a logic of proofs and refutations (Section~\ref{sec:nd}). Unlike existing proof systems for $\Nt$, $\BPR$ is {\em pure}, in the sense that each rule operates on a single connective. This design yields key proof-theoretic properties, including the subformula property. We also relate $\BPR$ to Nelson's logic, highlighting both the conceptual and technical consequences of our treatment of refutation.

We then show that $\BPR$ smoothly extends Wansing's natural deduction proof system $\BintN$ for $\Bint$~\cite{Wan110.1093/logcom/ext035}. However, unlike $\BintN$, a {\em direct proof} of (weak) normalisation for $\BPR$ is presented, and {\em consistency}\footnote{\cite{DBLP:journals/jsyml/Wansing99,Sano2025} study systems for nonmonotonic inference based on expansions of Nelson's logics $\Nf^\perp,\Nt^\perp$~\cite{NelsonConstructiveFalsity,Wansing1995-HEISNI,Odintsov2008-ODICNA} with Gabbay's consistency operator $\mathsf{M}$~\cite{gabbay82}. The connection with our approach is marginal since, in our system, consistency is guaranteed purely within the proof system, without adding extra operators.} is shown for both proofs and  refutations (Section~\ref{sec:pt}).

The semantics for $\BPR$ (Section~\ref{sec:bes}), formulated via base extension semantics ($\Bes$)~\cite{Sandqvist2015IL,sep-proof-theoretic-semantics,Ayhan-seq,DBLP:journals/logcom/dAragona25,Stafford}, makes explicit the idea of counterexample construction through atomic refutational rules, ensuring that no formula is simultaneously provable and refutable within the same atomic base. Furthermore, it retains all properties established in \cite{barrosonascimento2025bilateralbaseextensionsemantics,Ayhan-seq}, including the possibility of semantic embeddings via bilateral semantic harmony. Soundness and completeness results are shown, guaranteeing the correctness of the framework (Section~\ref{sec:sc}). 

We end this introduction by noting that there are extensive and well-developed discussions  in the literature on proof systems for Nelson's strong negation (see, \eg, \cite{Gurevich1977,Metcalfe2009-METASC,Kurbis-Petrukhin-2021,negri-sn} and refer to Appendix~\ref{sec:comp} for a more detailed discussion) and bilateralism (see, \eg, \cite{DBLP:journals/flap/Kurbis21,francez2014bilateralism,Gabbay2017BilateralismCritique,Francez2018ReplyToGabbay,Synthese-Simonelli,DBLP:journals/jphil/Simonelli24}). %The aim of this paper is to add some new insights to the debate. 
This paper aims at offering a new perspective that bridges these two strands.

\section{Bilateral Proof and Refutation ($\BPR$)}\label{sec:nd}
As noted in~\cite{DBLP:journals/japll/Wansing17}, inferential purity does not hold in some other natural deduction calculi for Nelson's logics with strong constructive negation~\cite{Gurevich1977,Kurbis-Petrukhin-2021,negri-sn}, in the sense that they display two connectives at a time in definitions (they also fail the subformula property), \eg:
\vspace{-0.2cm}
\begin{center}
$
\infer{\nn(\phi \to \psi)}{\phi & \nn\psi} \qquad\qquad \infer{\gamma}{\nn(\phi\wedge\psi)& \deduce{\gamma}{\deduce{\vspace{0.1cm}}{[\nn\phi]}} & \deduce{\gamma}{\deduce{\vspace{0.1cm}}{[\nn\psi]}}}
$
\end{center}
\vspace{-0.2cm}
Wansing's solution consisted of replacing the truth-preserving transitions to/from negated formulas by falsity preserving transitions to/from non-negated formulas, so the above rules are replaced by:
\vspace{-0.2cm}
\begin{center}
$
\infer={\phi \to \psi}{\infer{\phi}{} & \infer={\psi}{}} \qquad\qquad \infer{\gamma}{\infer={\phi\wedge\psi}{}& \infer{\gamma}{\llbracket\phi\rrbracket} & \infer{\gamma}{\llbracket\psi\rrbracket}}
$
\end{center}
\vspace{-0.2cm}
where the strong negation operator is substituted by double inference/discharge bars/brackets. The natural deduction system $\BintN$~\cite{DBLP:journals/japll/Wansing17} (Fig~\ref{fig:BintN}) formalises a conservative extension of propositional intuitionistic logic over a set of atomic variables $\At$, with connectives $\to, \wedge, \vee, \top, \bot$, obtained by adding co-implication ($\mapsfrom$), a connective dual to implication.  In $\BintN$, 
co-implication internalizes the preservation of {\em disproof} from the conclusion of a valid inference (understood as logical consequence) to its premises.
Hence the implication plays in {\em verificationism} a role dual to that played by co-implication in {\em falsificationism}. 

The system $\BintNs$~\cite{barrosonascimento2025bilateralbaseextensionsemantics} is a variant of $\BintN$. The key difference lies in the addition of the rules $E\vee(+)$, $E\wedge(-)$, $\bot(+)$, and $\top(-)$ (Fig.~\ref{fig:BintNs}), without which a direct proof of normalisation is not possible.
The bilateral natural deduction system $\BPR$ (Fig.~\ref{fig:SystemBPR}) then builds on $\BintNs$ by introducing the rule $PR$ (Fig.~\ref{fig:BPR}), which coordinates {\em proofs and refutations}. This single addition allows one to move from a paraconsistent to a consistent system, but it also introduces significant proof-theoretic challenges, as discussed in the following sections.

%As in $\BintN$, 
$\BPR$ features two kinds of derivations: those concluding with a formula under a single line (proofs) and those concluding under a double line (refutations). A proof (resp. a refutation) may contain refutations (resp. proofs) as subderivations. Dotted lines serve as placeholders for either single or double lines, with the proviso that in any application of $\lor(+)$ or $\land(-)$, all dotted lines must be replaced uniformly, either all by single lines or all by double lines.

The premises of a derivation in $\BPR$ are partitioned into two classes: assumptions, representing proofs, and contra-assumptions, representing refutations. Accordingly, the conclusion of a derivation is indexed by an ordered pair $(\Gamma; \Delta)$, where $\Gamma$ and $\Delta$ are finite sets of assumptions and contra-assumptions, respectively. Intuitively, $\Gamma$ collects hypotheses taken to be proved, whereas $\Delta$ collects those taken to be refuted. Finally, single (resp. double) square brackets $[\cdot]$ (resp. $\llbracket\cdot\rrbracket$) indicate the discharge of assumptions (resp. contra-assumptions).%, while double square brackets  indicate the discharge of .

We denote the existence of deductions through indexed turnstiles: $\Gamma ; \Delta \vdash^{+} \phi$ holds if and only if there is a deduction in $\BPR$ which has conclusion $\phi$, depends on the premises $(\Gamma'; \Delta')$ for $\Gamma' \subseteq \Gamma$ and $\Delta' \subseteq \Delta$ and ends with an application of a proof rule, and $\Gamma ; \Delta \vdash^{-} \phi$ holds if and only if there is such a deduction ending with an application of a refutation rule instead.

\begin{figure*}[htbp]
%\pagenumbering{gobble}
    \centering
    \scriptsize  % makes all proof trees smaller
    \setlength{\tabcolsep}{3pt}
    \renewcommand{\arraystretch}{0.8}

    % ===========================
    % (1) N2Int
    % ===========================
    \begin{subfigure}[t]{1\textwidth}
        \centering
        \vspace{-0.8em}
\begin{prooftree}
\AxiomC{$\Gamma, [\phi]; \Delta$}
\noLine
\UnaryInfC{$\Pi$}
\UnaryInfC{$\psi$}
\RightLabel{\tiny{$I \to (+)$}}
\UnaryInfC{$\phi \to \psi$}
\DisplayProof
\qquad
\AxiomC{$\Gamma_1 ; \Delta_1$}
\noLine
\UnaryInfC{$\Pi_1$}
\UnaryInfC{$\phi \to \psi$}
\AxiomC{$\Gamma_2; \Delta_{2}$}
\noLine
\UnaryInfC{$\Pi_2$}
\UnaryInfC{$\phi$}
\RightLabel{\tiny{$E \to (+)$}}
\BinaryInfC{$\psi$}
\DisplayProof
\qquad
\AxiomC{$\Gamma; \Delta$}
\noLine
\UnaryInfC{$\Pi$}
\UnaryInfC{$\phi$}
\RightLabel{\tiny{$I_{1} \vee (+)$}}
\UnaryInfC{$\phi \vee \psi$}
\end{prooftree}
\begin{prooftree}
\AxiomC{$\Gamma; \Delta$}
\noLine
\UnaryInfC{$\Pi$}
\UnaryInfC{$\psi$}
\RightLabel{\tiny{$I_{2} \vee (+)$}}
\UnaryInfC{$\phi \vee \psi$}
\DisplayProof
\qquad
\AxiomC{$\Gamma_1; \Delta_{1}$}
\noLine
\UnaryInfC{$\Pi_1$}
\UnaryInfC{$\phi$}
\AxiomC{$\Gamma_2; \Delta_{2}$}
\noLine
\UnaryInfC{$\Pi_2$}
\UnaryInfC{$\psi$}
\RightLabel{\tiny{$I \wedge (+)$}}
\BinaryInfC{$\phi \wedge \psi$}
\DisplayProof
\quad
\AxiomC{$\Gamma; \Delta$}
\noLine
\UnaryInfC{$\Pi$}
\UnaryInfC{$\phi \wedge \psi$}
\RightLabel{\tiny{$ E_{1} \wedge (+)$}}
\UnaryInfC{$\phi$}
\DisplayProof
\qquad
\AxiomC{$\Gamma; \Delta$}
\noLine
\UnaryInfC{$\Pi$}
\UnaryInfC{$\phi \wedge \psi$}
\RightLabel{\tiny{$E_{2} \wedge (+)$}}
\UnaryInfC{$\psi$}
\end{prooftree}

\begin{prooftree}
\AxiomC{$\Gamma_{1}; \Delta_{1}$}
\noLine
\UnaryInfC{$\Pi_{1}$}
\UnaryInfC{$\phi$}
\AxiomC{$\Gamma_{2}; \Delta_{2}$}
\noLine
\UnaryInfC{$\Pi_{2}$}
\doubleLine
\UnaryInfC{$\psi$}
\RightLabel{\tiny{$I \mapsfrom (+)$}}
\BinaryInfC{$\phi \mapsfrom \psi$}
\DisplayProof
\quad
\AxiomC{$\Gamma; \Delta$}
\noLine
\UnaryInfC{$\Pi$}
\UnaryInfC{$\phi \mapsfrom \psi$}
\RightLabel{\tiny{$ E_{1} \mapsfrom (+)$}}
\UnaryInfC{$\phi$}
\DisplayProof
\quad
\AxiomC{$\Gamma; \Delta$}
\noLine
\UnaryInfC{$\Pi$}
\UnaryInfC{$\phi \mapsfrom \psi$}
\doubleLine
\RightLabel{\tiny{$ E_{2} \mapsfrom (+)$}}
\UnaryInfC{$\psi$}
\end{prooftree}

\begin{prooftree}
\AxiomC{$\Gamma; \Delta, \llbracket \psi \rrbracket$}
\noLine
\UnaryInfC{$\Pi$}
\doubleLine
\UnaryInfC{$\phi$}
\doubleLine
\RightLabel{\tiny{$I \mapsfrom (-)$}}
\UnaryInfC{$\phi \mapsfrom \psi$}
\DisplayProof
\qquad
\AxiomC{$\Gamma_1 ; \Delta_1$}
\noLine
\UnaryInfC{$\Pi_1$}
\doubleLine
\UnaryInfC{$\phi \mapsfrom \psi$}
\AxiomC{$\Gamma_2; \Delta_{2}$}
\noLine
\UnaryInfC{$\Pi_2$}
\doubleLine
\UnaryInfC{$\psi$}
\RightLabel{\tiny{$E \mapsfrom (-)$}}
\doubleLine
\BinaryInfC{$\phi$}
\DisplayProof
\quad
\AxiomC{$\Gamma; \Delta$}
\noLine
\UnaryInfC{$\Pi$}
\doubleLine
\UnaryInfC{$\phi$}
\RightLabel{\tiny{$ I_{1} \wedge (-)$}}
\doubleLine
\UnaryInfC{$\phi \wedge \psi$}
\DisplayProof
\quad
\AxiomC{$\Gamma; \Delta$}
\noLine
\UnaryInfC{$\Pi$}
\doubleLine
\UnaryInfC{$\psi$}
\RightLabel{\tiny{$I_{2} \wedge (-)$}}
\doubleLine
\UnaryInfC{$\phi \wedge \psi$}
\end{prooftree}

\begin{prooftree}
\AxiomC{$\Gamma_1; \Delta_{1}$}
\noLine
\UnaryInfC{$\Pi_1$}
\doubleLine
\UnaryInfC{$\phi$}
\AxiomC{$\Gamma_2; \Delta_{2}$}
\noLine
\UnaryInfC{$\Pi_2$}
\doubleLine
\UnaryInfC{$\psi$}
\RightLabel{\tiny{$I \vee (-)$}}
\doubleLine
\BinaryInfC{$\phi \vee \psi$}
\DisplayProof
\quad
\AxiomC{$\Gamma; \Delta$}
\noLine
\UnaryInfC{$\Pi$}
\doubleLine
\UnaryInfC{$\phi \vee \psi$}
\RightLabel{\tiny{$E_{1} \vee (-)$}}
\doubleLine
\UnaryInfC{$\phi$}
\DisplayProof
\qquad
\AxiomC{$\Gamma; \Delta$}
\noLine
\UnaryInfC{$\Pi$}
\doubleLine
\UnaryInfC{$\phi \vee \psi$}
\RightLabel{\tiny{$E_{2} \vee (-)$}}
\doubleLine
\UnaryInfC{$\psi$}
\end{prooftree}

\begin{prooftree}
\AxiomC{$\Gamma_{1}; \Delta_{1}$}
\noLine
\UnaryInfC{$\Pi_{1}$}
\UnaryInfC{$\phi$}
\AxiomC{$\Gamma_{2}; \Delta_{2}$}
\noLine
\UnaryInfC{$\Pi_{2}$}
\doubleLine
\UnaryInfC{$\psi$}
\RightLabel{\tiny{$I\to (-)$}}
\doubleLine
\BinaryInfC{$\phi \to \psi$}
\DisplayProof
\quad
\AxiomC{$\Gamma; \Delta$}
\noLine
\UnaryInfC{$\Pi$}
\doubleLine
\UnaryInfC{$\phi \to \psi$}
\RightLabel{\tiny{$E_{1} \to (-)$}}
\UnaryInfC{$\phi$}
\DisplayProof
\quad
\AxiomC{$\Gamma; \Delta$}
\noLine
\UnaryInfC{$\Pi$}
\doubleLine
\UnaryInfC{$\phi \to \psi$}
\doubleLine
\RightLabel{\tiny{$E_{2} \to (-)$}}
\UnaryInfC{$\psi$}
\end{prooftree}

        \caption{System $\BintN$.}
        \label{fig:BintN}
    \end{subfigure}
    \hfill
    \\

   \begin{subfigure}[t]{1\textwidth}
        \centering

\begin{prooftree}
\quad
\AxiomC{$\Gamma_1; \Delta_{1}$}
\noLine
\UnaryInfC{$\Pi_1$}
\UnaryInfC{$\phi \vee \psi$}
\AxiomC{$\Gamma_2, [\phi]; \Delta_{2}$}
\noLine
\UnaryInfC{$\Pi_2$}
\dottedLine
\UnaryInfC{$\chi$}
\AxiomC{$\Gamma_3, [\psi]; \Delta_{3}$}
\noLine
\UnaryInfC{$\Pi_3$}
\dottedLine
\UnaryInfC{$\chi$}
\RightLabel{\tiny{$E\vee (+)$}}
\dottedLine
\TrinaryInfC{$\chi$}
\DisplayProof
\,
%\end{prooftree}
%
%\begin{prooftree}
\AxiomC{$\Gamma_1; \Delta_{1}$}
\noLine
\UnaryInfC{$\Pi_1$}
\doubleLine
\UnaryInfC{$\phi \wedge \psi$}
\AxiomC{$\Gamma_2; \Delta_{2},  \llbracket \phi \rrbracket$}
\noLine
\UnaryInfC{$\Pi_2$}
\dottedLine
\UnaryInfC{$\chi$}
\AxiomC{$\Gamma_3; \Delta_{3}, \llbracket \psi \rrbracket$}
\noLine
\UnaryInfC{$\Pi_3$}
\dottedLine
\UnaryInfC{$\chi$}
\RightLabel{\tiny{$ E \wedge (-)$}}
\dottedLine
\TrinaryInfC{$\chi$}
\end{prooftree}

\begin{prooftree}
%\AxiomC{}
%\noLine
%\UnaryInfC{}
%\noLine
%\UnaryInfC{}
%\noLine
%\UnaryInfC{}
%\noLine
%\UnaryInfC{}
%\noLine
%\UnaryInfC{}
%\RightLabel{\tiny{$ \top (+)$}}
%\UnaryInfC{$\top$}
%\DisplayProof
  \AxiomC{$\Gamma; \Delta$}
\noLine
\UnaryInfC{$\Pi$}
\UnaryInfC{$\bot$}
\RightLabel{\tiny{$ \bot(+)$}}
\dottedLine
\UnaryInfC{$\phi$}
\DisplayProof
\qquad
%\AxiomC{}
%\noLine
%\UnaryInfC{}
%\noLine
%\UnaryInfC{}
%\noLine
%\UnaryInfC{}
%\noLine
%\UnaryInfC{}
%\noLine
%\UnaryInfC{}
%\doubleLine
%\RightLabel{\tiny{$ \bot (-)$}}
%\UnaryInfC{$\bot$}
%\DisplayProof
%\qquad
    \AxiomC{$\Gamma; \Delta$}
\noLine
\UnaryInfC{$\Pi$}
\doubleLine
\UnaryInfC{$\top$}
\RightLabel{\tiny{$ \top (-)$}}
\dottedLine
\UnaryInfC{$\phi$}
\end{prooftree}

        \caption{Rules added to $\BintN$ in order to obtain the System $\BintNs$.}
        \label{fig:BintNs}
    \end{subfigure}
    \hfill
    \\
    
    % ===========================
    % (2) BPR
    % ===========================
    \begin{subfigure}[t]{1\textwidth}
        \centering
        %\vspace{-0.8em}

    \bigskip
        
\begin{prooftree}
\AxiomC{$\Gamma_1; \Delta_1$}
\noLine
\UnaryInfC{$\Pi_1$}
    \UnaryInfC{$\phi$}
    \AxiomC{$\Gamma_2; \Delta_2$}
    \noLine
\UnaryInfC{$\Pi_2$}
    \doubleLine
    \UnaryInfC{$\phi$}
    \dottedLine
    \RightLabel{\tiny{$PR$}}
    \BinaryInfC{$\psi$}
\end{prooftree}

        \caption{Proofs and refutations rule.}
        \label{fig:BPR}
    \end{subfigure}

    % ===========================
    % Global caption
    % ===========================
    \caption{System $\BPR$}
    \label{fig:SystemBPR}
\end{figure*}

%Before presenting the normalisation theorem, we comment on some of the design choices underlying our systems. In $\BintN$, the absence of the proof-theoretic and refutational versions of the rules mentioned above prevents direct proofs of normalisation and of the subformula property -- see Appendix~\ref{sec:comp} for a more detailed discussion. As shown in~\cite{Wan110.1093/logcom/ext035}, normalisation for $\BintN$ is instead established indirectly, by embedding its derivations into natural deduction for intuitionistic logic and then appealing to Prawitz's normalisation theorem~\cite{prawitz1965}. Similar modifications were introduced in~\cite{Ayhan-seq} to obtain cut elimination and the subformula property for the sequent system $\mathsf{SC2Int}$.

At this point, one may reasonably ask: how do these systems capture Nelson's strong negation, if the operator $\nn\ $  is not part of the grammar? The next result, presented in~\cite{oddsson2025strongnegationdefinable2int}, shows that extending $\BintN$ (and hence $\BPR$) with strong negation does not increase its expressive power. 
\begin{proposition}\label{ex:der}  Strong negation $\nn\phi$ is definable in $\BintN/\BPR$ as $(\phi \land (\phi \to (\phi \mapsfrom \phi))) \lor ((\phi \to \phi) \mapsfrom \phi)$. Moreover, the following rules are derivable in $\BintN/\BPR$: 
\vspace{-0.2cm}
\begin{prooftree}
    \AxiomC{$\Gamma; \Delta$}
\noLine
\UnaryInfC{$\Pi$}
    \doubleLine
    \UnaryInfC{$\phi$}
    \RightLabel{\tiny{$I \sim (+)$}}
    \UnaryInfC{$\sim \phi$}
    \DisplayProof
    \qquad
     \AxiomC{$\Gamma; \Delta$}
\noLine
\UnaryInfC{$\Pi$}
    \UnaryInfC{$\sim \phi$}
       \RightLabel{\tiny{$E \sim (+)$}}
    \doubleLine
    \UnaryInfC{$\phi$}
    \DisplayProof
    \qquad
       \AxiomC{$\Gamma; \Delta$}
\noLine
\UnaryInfC{$\Pi$}
    \UnaryInfC{$\phi$}
    \RightLabel{\tiny{$I \sim (-)$}}
    \doubleLine
    \UnaryInfC{$\sim \phi$}
    \DisplayProof
    \qquad
     \AxiomC{$\Gamma; \Delta$}
\noLine
\UnaryInfC{$\Pi$}
    \doubleLine
    \UnaryInfC{$\sim \phi$}
       \RightLabel{\tiny{$E \sim (-)$}}
    \UnaryInfC{$\phi$}
\end{prooftree}

\end{proposition}
Now, this brings important consequences.
In fact, the use of a framework in which we have both a weak falsity operator $\neg \phi$ and a concept of refutation equivalent to the strong falsity operator $\nn\phi$ also allows the study of their interactions in more detail. In $\BPR$ we can easily conclude that from a refutation of $\phi$ follows a constructive proof of the negation $\phi\to\bot$, as well as that from a proof of $\phi$ follows a refutation of the co-negation $\top\mapsfrom\phi$: %\footnote{Observe that these inferences are not valid in $\Bint$. In fact, it is straightforward to use the semantic framework in \cite{barrosonascimento2025bilateralbaseextensionsemantics} to provide a counterexample in $\Bint$: If a base $\mathcal{B}$ contains a proof rule concluding $p$ from no premises, a refutation rule concluding $p$ from no premises and no other rules, it is easy to check that $\Vdash^{+}_{\mathcal{B}} p$ and $\Vdash^{-}_{\mathcal{B}} p$ but $\nVdash^{+}_{\mathcal{B}} p \to \bot$ and $\nVdash^{-}_{\mathcal{B}} \top \to p$. This is not a counterexample in our semantics for $\BPR$ since $\mathcal{B}$ would be inconsistent.}: %(represented by $- \phi$):
\vspace{-0.2cm}
\begin{prooftree}
    \AxiomC{}
    \RightLabel{\scriptsize{$1$}}
    \UnaryInfC{$\phi$}
    \AxiomC{}
    \doubleLine
    \UnaryInfC{$\phi$}
    \RightLabel{\scriptsize{PR(+)}}
    \BinaryInfC{$\bot$}
    \RightLabel{$1$}
    \UnaryInfC{$\phi \to \bot$}
    \DisplayProof
    \qquad
    \qquad
        \AxiomC{}
    \UnaryInfC{$\phi$}
    \AxiomC{}
    \doubleLine
        \RightLabel{\scriptsize{$1$}}
    \UnaryInfC{$\phi$}
    \doubleLine
      \RightLabel{\scriptsize{PR(-)}}
    \BinaryInfC{$\top$}
    \RightLabel{$1$}
        \doubleLine
    \UnaryInfC{$\top \mapsfrom \phi$}
\end{prooftree}
\vspace{-0.1cm}
%These deductions  due to their use of $PR(+)$ and $PR(-)$. 
%From the semantics of $\to$ and $\mapsfrom$ together with the consistency constraints, i
It follows that there is a (constructive) proof of $\phi \to \bot$ iff there is no proof of $\phi$ and a (constructive) refutation of $\top \mapsfrom \phi$ iff there is no refutation of $\phi$. Hence those deductions are essentially showing that {\em strong falsity} implies {\em weak falsity} and {\em strong provability} implies {\em weak provability} (that is, absence of refutations). 
Note that these inferences are not valid in $\BintN$--this is a feature specific to $\BPR$.

At the end of Section~\ref{sec:sc} we prove that the converse does not hold; that is, weak falsity does not imply strong falsity in $\BPR$.

\section{Normalisation and the subformula principle}\label{sec:pt}

%\subsection{Normalisation and the subformula property}

In this section we prove the normalisation theorem for  $\BPR$. The aim is to show that if there exists a proof (respectively, a refutation) of a formula 
$\phi$ from a finite set of assumptions $\Theta$ and contra-assumptions $\Delta$ (that is, a set of premises $\Gamma = \Delta \cup \Theta$), then there exists a derivation of the same kind in {\em normal form}.

The proof of the normalisation theorem yields a systematic procedure for eliminating detours, that is, redundant segments of derivations that do not contribute to their essential logical structure. In standard natural deduction, such detours arise from consecutive introduction and elimination steps for the same connective. In our setting, however, they also encompass {\em non-trivial} proof-refutation interactions and other mixed rule combinations, revealing a richer and more intricate notion of redundancy.

Once normalisation is established, we prove that the subformula property holds: every formula occurring in a normal derivation is a subformula of either an undischarged assumption or the final conclusion. This result is not evident, as it usually does not follow directly from the normalisation result. 
Indeed, consider the following derivation of $\top$ from $\bot ; \emptyset$ in $\BintN$:
\vspace{-0.4cm}
\begin{prooftree}
\AxiomC{}
\UnaryInfC{$\bot$}
\RightLabel{\tiny{$\bot(+)$}}
\UnaryInfC{$\top \mapsfrom \top$}
\doubleLine
\RightLabel{\tiny{$E_{2} \mapsfrom (+)$}}
\UnaryInfC{$\top$}
\end{prooftree}
\vspace{-0.2cm}
Since $\top \mapsfrom \top$ is neither a subformula of any open premise nor of the conclusion, this derivation violates the subformula property, even though a different form of normalisation still holds for $\BintN$\footnote{In fact, normalisation for $\BintN$ is established indirectly, by embedding its derivations into natural deduction for intuitionistic logic and then appealing to Prawitz's normalisation theorem~\cite{prawitz1965}. On the other hand, in~\cite{Ayhan-seq}  cut-elimination and the subformula property are proved directly for the {\em sequent system} $\mathsf{SC2Int}$.}.

Finally, as a further consequence, we obtain consistency: $\bot$ is not provable and $\top$ is not refutable in $\BPR$.

\begin{definition}
%    The {\em degree} of a formula is defined as the number of logical connectives different from $\bot$ and $\top$ occurring on it.
The {\em degree of a formula $\phi$}, denoted $d[\phi]$, is defined inductively by
$$
\begin{array}{lcll}
d[p] &=& 0 & \mbox{ for }  p\in\At\\
d[\phi \circ \psi] &=& d[\phi] + d[\psi] + 1 &\mbox{ for } \circ\in\{\to, \mapsfrom,\vee,\wedge\}
\end{array}
$$
\end{definition}

\begin{definition}\label{def:complexityvalueatoms} Let $\Pi$ be a derivation in $\BPR$, $d$ be the greatest degree of a formula that occurs as premise or conclusion of an instance of $PR$, $\bot (+)$ or $\top (-)$ in $\Pi$, $n$ the number of occurrences of formulas with degree $d$ as premises of applications of $PR$, $\bot (+)$ or $\top (-)$ in $\Pi$, and $m$ the number of occurrences of formulas with degree $d$ as conclusions of those same rules. 
The {\em complexity value} of $\Pi$ is defined as the pair $\langle d, n + m \rangle$. We say that $\langle d',  n' + m' \rangle <  \langle d, n+m \rangle$ holds if and only if either $d' < d$ or both $d' = d$ and $n'+m' < n +m $.
\end{definition}
Notice that, by defining the second component of the complexity value as a sum, %we make it so that 
formula occurrences that are simultaneously a premise and a conclusion of the relevant rules are counted twice -- this is essential for proving Lemma~\ref{atomicpremises}.
%with our rules for assertion of disjunction and refutation of conjunction (which are the traditional ones). 
%This is not necessary in the similar proof for $\mathbf{HB_{2}}$ \cite[pg. 396-397]{Valle-Inclan2023-VALHAN}, but only because the rules for such connectives are changed.}

The following lemmas play a crucial role in the normalisation results based on Prawitz's techniques. They establish that certain rule applications can be restricted to the atomic case and that consecutive applications of them can be deleted. The proofs are given in the Appendix~\ref{app:sec3}.

\begin{lemma}\label{atomicpremises}
    Any deduction $\Pi$ with conclusion $\phi$ and premises $\Gamma$ can be reduced to a deduction $\Pi'$ with conclusion $\phi$ and premises $\Gamma$ such that all instances of $PR$, $\bot (+)$ and $\top (-)$  have premises and conclusions in $\At$.
\end{lemma}

\begin{lemma}\label{MariaLemma}
    Any deduction $\Pi$ showing $\Gamma; \Delta \vdash^{*} \phi$ can be reduced to a deduction $\Pi'$ showing $\Gamma' ; \Delta' \vdash^{*} \phi$ for $\Gamma' \subseteq \Gamma$ and $\Delta' \subseteq \Delta$ such that no formula occurrence is both the conclusion of an application of $\bot (+)$, $\top(-)$ or $PR$ and the premise of an application of $\bot (+)$, $\top(-)$ or $PR$.
\end{lemma}

An interesting feature of the proof of Lemma \ref{atomicpremises} is that it essentially relies on the duality (or harmony) of proof and refutation rules for the logical operators. 
%\cyan{It is also  a variant of the proof of the same result for $\mathbf{HB_{2}}$ in \cite[pg. 396-397]{Valle-Inclan2023-VALHAN}. The main difference between both proofs is that $\BPR$, unlike $\mathbf{HB_{2}}$, uses the traditional rules for assertion of disjunction and refutation of conjunction, which makes it necessary to change the inductive basis (cf. Definition \ref{def:complexityvalueatoms}) and the proof steps.} 
A proof-theoretic analysis of this bilateral concept of harmony, which could in principle be used to extend the result to other intuitionistic bilateral logics\footnote{Those rules do not always preserve classical principles, as shown in \cite{FernandoFerreiraBilateral}.}, is carried out in \cite{francez2014bilateralism}. A semantic counterpart of bilateral harmony can also be found in the weak and strong semantic harmony results proved in \cite{barrosonascimento2025bilateralbaseextensionsemantics}.

Our normalisation proof adapts Prawitz's strategy for intuitionistic logic~\cite{prawitz1965} to the bilateralist setting. Some modifications are required, however, since the $PR$ rule introduces new complications that may threaten the subformula property. To illustrate this, consider the following deduction: 
%Our normalisation proof adapts Prawitz's strategy for intuitionistic logic \cite{prawitz1965} to the bilateralist setting. The proof has to be changed slighly because the $PR (+)$ and $PR (-)$ rules bring about new difficulties which can affect satisfaction of the subformula principle. The main change consists in the fact that we first apply reductions so as to make the premises and conclusions of any application of $\bot (+)$, $\top (-)$, $PR (+)$ and $PR (-)$ atomic in order to obtain deductions in what we call pre-normal form and only then reduce the deductions to normal form. We also adapt the inductive proof by showing the result through induction on the degree of pre-normal form deductions instead of the starting deductions. Unlike in our definitions, Prawitz lets maximal segments start with application of the rule we call $\bot (+)$ and only reduce their degrees when the consequence of the application occurs in a maximal segment. This change of the inductive structure is optional, as we could also obtain a thoroughly Prawitz-style proof by allowing maximal segments to start with conclusions of $PR (+)$, $PR(-)$, $\bot (+)$ and $\top (-)$. However, we would also need to allow the last formula occurrence of a maximal segment to be a premise of an application of $PR (+)$ or $PR (-)$ if we want to prove the subformula principle. To see why, consider the following deduction:
\vspace{-0.4cm}
\begin{prooftree}
    \AxiomC{}
    \UnaryInfC{$\phi$}
    \AxiomC{}
    \UnaryInfC{$\psi$}
    \RightLabel{\scriptsize{$\land I (+)$}}
    \BinaryInfC{$\phi \land \psi$}
    \AxiomC{}
    \doubleLine
    \UnaryInfC{$\phi$}
    \doubleLine
       \RightLabel{\scriptsize{$\land I (-)$}}
    \UnaryInfC{$\phi \land \psi$}
      \RightLabel{\scriptsize{$PR$}}
    \BinaryInfC{$\chi$}
\end{prooftree}
\vspace{-0.2cm}
Since $\phi \land \psi$ is a subformula of neither the premises nor the conclusion of this deduction, establishing the subformula property requires reductions specific to the bilateralist setting. Such reductions must show either that every deduction can be transformed into one in which no premise of a $PR$ application is the conclusion of an introduction rule, or--following the approach we adopt--that all such premises can be reduced to atomic form.

The main difference in our normalisation proof is that we first apply reductions so as to make the premises and conclusions of any application of $\bot (+)$, $\top (-)$, $PR$ {\em atomic} in order to obtain deductions in what we call atomic normal form, then reducing the deductions to normal form and simplifying them further.

The following definitions are standard~\cite{prawitz1965}.

\begin{definition}
    A {\em maximal segment} in a deduction $\Pi$ is a sequence $\phi_{1}, \ldots, \phi_{n}$ with at least one element such that $\phi_{1}$ is the conclusion of a $I$-rule, $\phi_{n}$ is the major premise of an $E$-rule and, for all $1 \leq i < n$ (if any), $\phi_{i}$ is the minor premise of $\lor (+)$ or $\land (-)$.
    
\noindent
A maximal segment $\alpha$ is said to be {\em above a formula occurrence $\beta$ in a deduction} iff all formula occurrences in $\alpha$ are above the formula occurrence $\beta$.

%\end{definition}
%As usual, notice that $\phi_{1}, \ldots , \phi_{n}$ are all formula occurrences of the same shape. \Elaine{Did not get this}
\noindent
%\begin{definition}
The {\em degree} of a maximal segment is the degree of the formula $\phi_{1}$. The {\em length} of a maximal segment is the number of formula occurrences in it.
\end{definition}

  %  The {\em degree} of a maximal segment in a deduction $\Pi$ is a pair $\langle d, l  \rangle$ such that $d$ is the degree of $\phi_{1}$ and $l$ is the length of the segment. We say that $\langle d, l \rangle < \langle d', l' \rangle$ whenever $d < d'$ or $d = d'$ and $l < l'$.

    %$d$ is called the segment's {\em complexity}, $l$  We also define $\langle d, l \rangle < \langle d', l' \rangle$ as holding whenever either $d < d'$ or $d = d'$ and $l < l'$.

\begin{definition}
    A deduction is in {\em atomic form} if all conclusions and premises of applications of $\bot (+)$, $\top (-)$ and $PR$ have degree $0$.
\end{definition}

\begin{definition}
The {\em inductive value} of a deduction $\Pi$ is a pair $\langle d, s \rangle$, where $m$ is the highest degree of a maximal segment in $\Pi$ and $s$ is the sum of the length of all maximal segments with degree $d$. We define $\langle d , s \rangle < \langle d', s' \rangle$ as holding whenever $d < d'$ or $d = d'$ and $s < s'$.

\noindent
A deduction is in {\em normal form} if it is in atomic form with inductive value $\langle 0, 0 \rangle$.
%\end{definition}
\noindent
%\begin{definition}
A deduction is in {\em simplified normal form} if it is in normal form and no formula occurrence is both the premise of an application of $PR$, $\bot (+)$ and $\top (-)$ and a conclusion of a rule with one of those shapes.
\end{definition}
We are now ready to establish normalisation for $\BPR$; the full proof is given in Appendix~\ref{app:sec3}. Our strategy proceeds in three stages. First, we transform the initial deduction into atomic form using the crucial Lemma~\ref{atomicpremises}. Next, we obtain a normal deduction by adapting Prawitz's normalisation strategy to the bilateral setting. Finally, we simplify the resulting deduction by means of Lemma~\ref{MariaLemma}. These simplifications are {\em non trivial} and required in order to secure the subformula property. And this order is {\em essential}: while normalisation procedures preserve atomic form and simplifications preserve normal form, reductions to atomic form may generate new maximal segments, and normalisation may introduce new formula occurrences that must subsequently be eliminated by simplification.

\begin{theorem}\label{thm:normalisation}
    Every deduction $\Pi$  showing $\Gamma ; \Delta \vdash^{*} \phi$ can be reduced to a deduction showing $\Gamma ; \Delta \vdash^{*} \phi$ which is in simplified normal form.
\end{theorem}
The subformula property follows immediately from our normalisation proof.
\begin{corollary}\label{cor:sub}
 $\BPR$ satisfies the subformula property. Since $\BPR$ is a conservative extension of $\BintNs$, the latter system also satisfies the subformula property.
\end{corollary}

The proof of the following result %and proofs are straightforward 
is an adaptation of the one presented in~\cite{EcumenicalPTS} on atomic normalisation. 

%The proof of the following result is and straightforward adaptation of the proof of~\cite[Lemma 30]{EcumenicalPTS}.
\begin{lemma}\label{lemma:undischarged}
    If $\Pi$ is a normal derivation in $\BPR$ that does not end in an application of an $I$-rule, then $\Pi$ contains at least one undischarged hypothesis.
\end{lemma}

Remarkably, we can provide a syntactic proof of consistency for $\BPR$. This result plays a key role in the semantics introduced in the next section, which relies in particular on consistency proofs to establish completeness.

%\begin{proof} Straightforward adaptation of the proof of Lemma 30 in \cite{EcumenicalPTS}.

%If a proof does not end with and $I$-rule then it either ends with and E-rule or with an application of $PR(+)$, $PR(-)$, $\bot (+)$ or $\top (-)$. Let it be one  A consequence from our use of pre-normal form deductions for the normalisation proof is that we
%\end{proof}

\begin{corollary}\label{cor:consistencyBPR}
    $\BPR$ is consistent.
\end{corollary}
\begin{proof}
    Assume %, for the sake of reaching a contradiction, 
    that there is a derivation $\Pi$ showing 
    of $\vdash^+\bot$. Since no $I$-rule has a conclusion %with shape 
    $\bot$, the deduction cannot end with an application of such rules, so from Lemma~\ref{lemma:undischarged} it follows that there is at least one undischarged premise in the deduction, contradicting the assumption that it is a deduction showing $\vdash^+\bot$. Similarly for the case $\vdash^-\top$.
   % Then, by Theorem \ref{thm:normalisation}, $\Pi$ can be reduced to a normal derivation $\Pi'$ showing $\vdash^+_{\BPR}\bot$. No $I$-rules are capable of having $\bot$ as its conclusion, so 
    %By Theorem \ref{thm:shapenormal}, we have that there may not be any $I$-rules in $\Pi'$, as no such rule introduces $\bot$. But then, by Lemma \ref{lemma:undischarged}, $\Pi'$ contains at least one discharged proof assumption or contra-assumption, thus it cannot be a derivation showing $\vdash^+_{\BPR}\bot$. Hence, $\not\vdash^+_{\BPR}\bot$. A similar argument can be made to show that $\not\vdash^-_{\BPR}\top$.
\end{proof}

\section{Base-extension semantics}\label{sec:bes}
In this section we provide a base-extension semantics ($\Bes$) for $\BPR$. This choice of semantics is motivated by the fact that $\Bes$ enables non-trivial extensions of the results in~\cite{barrosonascimento2025bilateralbaseextensionsemantics} by combining techniques from~\cite{EcumenicalPTS} with new concepts. These must be integrated with our proof-theoretical results and have no direct counterpart in standard model-theoretic semantics.
Moreover, given that this is a paper in proof theory, we find it natural and conceptually appealing to adopt a semantic framework where proof theory is the main star.

Proof-theoretic semantics~\cite{pts-91,schroeder2006validity,sep-proof-theoretic-semantics} ($\Pts$) provides an alternative perspective for the meaning of logical operators compared to the viewpoint offered by model-theoretic semantics. In $\Pts$, the concept of {\em truth} is substituted for that of {\em proof}, emphasizing the fundamental nature of proofs as a means through which we gain demonstrative knowledge, particularly in mathematical contexts. This makes $\Pts$ a superior approach for comprehending reasoning since it ensures that the meaning of logical operators, such as connectives in logics, is defined based on their usage in inferences. 

$\Bes$~\cite{Sandqvist2015IL} is a strand of $\Pts$ where proof-theoretic validity is defined relative to a given collection $\mathcal{B}$ of inference rules defined over basic formulas of the language. 
Indeed, validity w.r.t. a set $\mathcal{B}$ of atomic rules has the general  shape
\vspace{-0.3cm}
\begin{center}
$
\Vdash_{\mathcal{B}} p \qquad \mbox{iff} \qquad\vdash_{\mathcal{B}} p
$
\end{center}
\vspace{-0.2cm}
where $\vdash_{\mathcal{B}} p$ indicates that $p$ is {\em derivable} in the proof system determined by $\mathcal{B}$. After defining validity for atoms one can also define validity for logical connectives via semantic clauses that express proof conditions (\eg, $A \land B$ is provable from $\mathcal{B}$ if and only if both $A$ and $B$ are provable in $\mathcal{B}$), which results in a framework that evaluates propositions exclusively in terms of proofs of its constituents.

%Our semantics is a non-trivial combination of the bilateral framework of \cite{barrosonascimento2025bilateralbaseextensionsemantics} with the consistency constraint from \cite{EcumenicalPTS}, augmenting this combination with an extra condition to secure epistemic consistency. 

%The conceptual groundwork for this semantics is contained mostly in \cite{Sandqvist2015IL}; directly related semantic concepts are also dealt with in \cite{Piecha2016,piecha2015failure,piecha2016completeness}. 
%Our semantics require use of some modifications presented in \cite{EcumenicalPTS} concerning the use of bases with consistency constraints, without which our essential Theorem~\ref{theorem:positivenegativesupport} would not hold. Broadly speaking, we combine the bilateral semantics presented in \cite{barrosonascimento2025bilateralbaseextensionsemantics} with the modifications presented in \cite{EcumenicalPTS} plus a few novel ones. Incidentally, a conceptual introduction to $\Bes$ which contextualizes the difference introduced by the consistency constraint is presented in \cite{EcumenicalPTS}, whereas a introduction to $\Bes$ which relates it to logical bilateralism is presented in \cite{barrosonascimento2025bilateralbaseextensionsemantics}.

\subsection{Basic definitions}

We adopt Sandqvist's~\cite{Sandqvist2015IL} terminology with the adaptations to the bilateral setting presented in \cite{barrosonascimento2025bilateralbaseextensionsemantics}. Our definitions use systems containing natural deduction rules over basic sentences for the semantical analysis, and we allow inference rules to discharge sets of basic hypotheses. On the other hand, we allow the logical constants $\bot$ and $\top$ to be manipulated by the rules as if they were atoms.

The propositional {\em base language}, denoted by $\At$--abusing the concept of atom-- is assumed to have a set $\{p_1, p_2,\ldots\}$ of countably many atomic propositions together with $\bot$ and $\top$. The elements of $\At$ are called {\em basic sentences}, or  {\em atoms}.

%Formulas are built from atoms using the binary connectives $\to,\wedge,\vee$ and the unit $\bot$. We use $\neg A$ as an abbreviation for $A \to\bot$, and 

\begin{definition}\label{def:bilateralbase}
A {\em bilateral atomic system} (a.k.a. {\em bilateral base}, or {\em base}) $\mathcal{B}$ is a (possibly empty) set of atomic rules of the form
\begin{prooftree}
\AxiomC{$[\Gamma_{\At}^{1}]; \llbracket \Delta_{\At}^{1} \rrbracket$}
\noLine
\UnaryInfC{.}
\noLine
\UnaryInfC{.}
\noLine
\UnaryInfC{.}
\UnaryInfC{$p_1$}
\AxiomC{$\ldots$}
\AxiomC{$[\Gamma^{n}_{\At}]; \llbracket \Delta^{n}_{\At} \rrbracket$}
\noLine
\UnaryInfC{.}
\noLine
\UnaryInfC{.}
\noLine
\UnaryInfC{.}
\doubleLine
\UnaryInfC{$p_n$}
\TrinaryInfC{$p$}
\DisplayProof
\AxiomC{$[\Gamma_{\At}^{1}]; \llbracket \Delta_{\At}^{1} \rrbracket$}
\noLine
\UnaryInfC{.}
\noLine
\UnaryInfC{.}
\noLine
\UnaryInfC{.}
\UnaryInfC{$p_1$}
\AxiomC{$\ldots$}
\AxiomC{$[\Gamma^{n}_{\At}]; \llbracket \Delta_{\At}^{n} \rrbracket$}
\noLine
\UnaryInfC{.}
\noLine
\UnaryInfC{.}
\noLine
\UnaryInfC{.}
\doubleLine
\UnaryInfC{$p_n$}
\doubleLine
\TrinaryInfC{$p$}
\end{prooftree}
where $p_i,p\in\At$, $\Gamma_{\At}^{i}$ is a (possibly empty) set of atomic proof assumptions and  $\Delta^{i}_{\At}$ is a (possibly empty) set of atomic contra-assumptions. The sequence $\langle p^{1}, ... , p^{n}\rangle$ of premises in a rule can be empty -- in this case the rule is called an {\em atomic axiom}. The sets $\Gamma^{i}_{\At}, \Delta^{i}_{\At}$ can be omitted when empty. The atom $p$ is the rule's {\em conclusion}.
\end{definition}

\begin{definition} \label{def:proofrulesrefutrulesandaxioms}
    An atomic rule is called a {\em proof rule} if a single line appears above its conclusion, and a {\em refutation rule} if a double line appear above it instead. Atomic axioms may also be called {\em proof axioms} if they are proof rules and {\em refutation axioms} otherwise.
\end{definition}

\begin{definition}\label{def:exts2}
An atomic system $\mathcal{C}$ is an {\em extension} of an atomic system $\mathcal{B}$ (written $\mathcal{C} \supseteq \mathcal{B}$) if $\mathcal{C}$ results from adding a (possibly empty) set of atomic rules to $\mathcal{B}$.   
\end{definition}

%If $S$ is an atomic system and $K$ is a set of sentence letters, we write $\overline{S(K)}$ for the {\em closure} of $K$ under all the rules in $S$ in the usual set-theoretic sense. An atomic system $S$ is {\em dense} if $\overline{S(\emptyset)}=\At$.

%\begin{definition}[Proof deductions]\label{def:deducatomicproofs} The consequence $\Gamma_{\At}; \Delta_{\At} \vdash^{+}_\mathcal{B} p$ holds if and only if there is a natural deduction derivation with conclusion $p$ depending on the open atomic proof assumptions $\Theta_{\At}$ and open atomic refutation assumptions $\Sigma_{\At}$  (for some $\Theta_{\At} \subseteq \Gamma_{\At}$ and $\Sigma_{\At} \subseteq \Delta_{\At}$) that only uses the rules of $\mathcal{B}$, which either consists in a single occurrence of a proof assumption or whose last rule application is a proof rule.
%\end{definition}

%\begin{definition}[Refutation deductions]\label{def:deducatomicrefs} The consequence $\Gamma_{\At}; \Delta_{\At} \vdash^{-}_\mathcal{B} p$ holds if and only if there is a natural deduction derivation with conclusion $p$ depending on the open atomic proof assumptions $\Theta_{\At}$ and open atomic refutation assumptions $\Sigma_{\At}$  (for some $\Theta_{\At} \subseteq \Gamma_{\At}$ and $\Sigma_{\At} \subseteq \Delta_{\At}$) that only uses the rules of $\mathcal{B}$, which either consists in a single occurrence of a refutation assumption or whose last rule application is a refutation rule.
%\end{definition}

\begin{definition} An {\em atomic proof} in the base $\mathcal{B}$ is a deduction using only the rules $\mathcal{B}$ which has no undischarged assumptions or contra-assumptions and ends with an application of a proof rule. We denote by $\vdash^{+}_\mathcal{B} p$ the existence of an atomic proof of $p$ in $\mathcal{B}$.

%\end{definition}
%\begin{definition} 
An {\em atomic refutation} in the base $\mathcal{B}$ is a deduction using only the rules $\mathcal{B}$ which has no undischarged assumptions or contra-assumptions and ends with an application of a refutation rule. We denote by $\vdash^{+}_\mathcal{B} p$ the existence of an atomic proof of $p$ in $\mathcal{B}$.
    
\end{definition}

%Notice that if  $p \in \{p_{1}, ... , p_{n}\}$ then $p_1, ..., p_{n}; \Gamma_{\At} \vdash^{+}_{\mathcal{B}} p$ and $\Gamma_{\At} ;  \{p_1, ..., p_{n} \}  \vdash^{-}_{\mathcal{B}} p$ regardless of the $\mathcal{B}$ and $\Gamma_{\At}$ due to how deducibility is defined.

Now we define new notions to be used in the new setting.

\begin{definition}\label{def:deduc}
An atomic system $\B$ is 
%\vspace{-0.1cm}
\begin{itemize} 
\item {\em logically consistent} if $\nvdash^{+}_{\B} \bot$ and $\nvdash^{-}_{\mathcal{B}} \top$.
\item {\em unit complete} if it contains a proof axiom concluding $\top$ and a refutation axiom concluding $\bot.$
\item  {\em epistemically consistent} if it is closed under rules with the following shape for all $p \in \At$\footnote{Technically speaking, our theorem requiring consistency only need bases to be closed under instances of either of the rules pictured above, but we include both in order to give precedence neither to the concept of proof nor of refutation.}:
\vspace{-0.4cm}
\begin{prooftree}
    \AxiomC{}
    \UnaryInfC{$p$}
    \AxiomC{}
    \doubleLine
    \UnaryInfC{$p$}
    \BinaryInfC{$\bot$}
    \DisplayProof
    \qquad \qquad
        \AxiomC{}
    \UnaryInfC{$p$}
    \AxiomC{}
    \doubleLine
    \UnaryInfC{$p$}
    \doubleLine
    \doubleLine
    \BinaryInfC{$\top$}
\end{prooftree}
\vspace{-0.1cm}
\item  {\em epistemically adequate} if it is logically consistent, unit complete and epistemically consistent.
\end{itemize}
\end{definition}

In the present paper we assume that bases and their extensions are always epistemically adequate (unless otherwise stated).

\subsection{Semantic clauses and epistemic consistency} \label{sec:semantics}

The semantics considered here are based on the clauses introduced in \cite{barrosonascimento2025bilateralbaseextensionsemantics}, except that we extend the domain $\At$ so as to include both $\bot$ and $\top$. %, thus falling into the scope of $(\At +)$ and $(\At -)$ instead of having dedicated clauses.

\begin{definition}\label{def:bilateralvalidity}
 \normalfont   The relations $\Vdash_{\mathcal{B}}^{+}, \Vdash_{\mathcal{B}}^{-}, \Vdash^{+} and \Vdash^{-}$ are defined in Fig~\ref{fig:sc}.
 An asterisk $*$ may be used as a placeholder for $+$ and $-$ when a result is to be shown for both signs. Whenever an asterisk is used in a statement or proof, all its instances must be uniformly replaced by the same sign.
\end{definition}
 %as follows:
\begin{figure}[t]
\begin{enumerate}
    \item[\namedlabel{c:at+}{($\At +$)}] $\Vdash_{\mathcal{B}}^{+} p$ iff $\vdash^{+}_{\mathcal{B}} p$, for $p \in \At$;

\medskip
    
        \item[\namedlabel{c:at-}{$(\At -)$}] $\Vdash_{\mathcal{B}}^{-} p$ iff $\vdash^{-}_{\mathcal{B}} p$, for $p \in \At$;

\medskip 

\item[\namedlabel{c:and+}{($\land +$)}] $\Vdash_{\mathcal{B}}^{+} \phi \land \psi$ iff $\Vdash_{\mathcal{B}}^{+} \phi$ and $\Vdash_{\mathcal{B}}^{+} \psi$

\medskip

\item[\namedlabel{c:and-}{($\land -$)}] $\Vdash_{\mathcal{B}}^{-} \phi \land \psi$ iff  $ \forall \mathcal{C} (\mathcal{C} \supseteq \mathcal{B})$ and $\forall p (p \in \At)$ : ($ \emptyset ; \phi \Vdash_{\mathcal{C}}^{+} p$ and $\emptyset ;\psi \Vdash_{\mathcal{C}}^{+} p$ implies $\Vdash_{\mathcal{C}}^{+} p$) and ($ \emptyset ; \phi \Vdash_{\mathcal{C}}^{-} p$ and $\emptyset ;\psi \Vdash_{\mathcal{C}}^{-} p$ implies $\Vdash_{\mathcal{C}}^{-} p$);

\medskip

\item[\namedlabel{c:or+}{($\lor +$)}] $\Vdash_{\mathcal{B}}^{+} \phi \lor \psi$ iff  $ \forall \mathcal{C} (\mathcal{C} \supseteq \mathcal{B})$ and $\forall p (p \in \At)$ $:  (\phi ;\emptyset \Vdash_{\mathcal{C}}^{+} p$ and $\psi ; \emptyset \Vdash_{\mathcal{C}}^{+} p$ implies $\Vdash_{\mathcal{C}}^{+} p)$ and $(\phi ;\emptyset \Vdash_{\mathcal{C}}^{-} p$ and $\psi ; \emptyset \Vdash_{\mathcal{C}}^{-} p$ implies $\Vdash_{\mathcal{C}}^{-} p)$;
\medskip

\item[\namedlabel{c:or-}{($\lor -$)}] $\Vdash_{\mathcal{B}}^{-} \phi \lor \psi$ iff $\Vdash_{\mathcal{B}}^{-} \phi$ and $\Vdash_{\mathcal{B}}^{-} \psi$;

\medskip

\item[\namedlabel{c:to+}{($\to +$)}] $\Vdash_\mathcal{B}^{+} \phi \to \psi$ iff $\phi ; \emptyset \Vdash_\mathcal{B}^{+} \psi$;

\medskip

\item[\namedlabel{c:to-}{($\to -$)}] $\Vdash_\mathcal{B}^{-} \phi \to \psi$ iff $\Vdash_\mathcal{B}^{+} \phi$ and $\Vdash_\mathcal{B}^{-} \psi$;

\medskip

\item[\namedlabel{c:mapsfrom+}{($\mapsfrom +$)}] $\Vdash_\mathcal{B}^{+} \phi \mapsfrom \psi$ iff $\Vdash_\mathcal{B}^{+} \phi$ and $\Vdash_\mathcal{B}^{-} \psi$;

\medskip

\item[\namedlabel{c:mapsfrom-}{($\mapsfrom -$)}] $\Vdash_\mathcal{B}^{-} \phi \mapsfrom \psi$ iff $\emptyset ; \psi \Vdash_\mathcal{B}^{-} \phi$;

\medskip

%\medskip

%\item For non-empty $\Gamma$, $\Gamma \Vdash_\mathcal{B}^{+} \phi$ iff $\forall (\mathcal{C} \supseteq \mathcal{B}) : \   \Vdash_{\mathcal{C}}^{+} \psi$ for all $\psi \in \Gamma$ implies $\Vdash_{\mathcal{C}}^{+} \phi$.

%\medskip

%\item For non-empty $\Gamma$, $\Gamma \Vdash_\mathcal{B}^{-} \phi$ iff $\forall \mathcal{C} (\mathcal{B} \supseteq \mathcal{C}) : \   \Vdash_{\mathcal{C}}^{-} \psi$ for all $\psi \in \Gamma$ implies $\Vdash_{\mathcal{C}}^{-} \phi$.

\item[\namedlabel{c:inf+}{(Inf $+$)}] For non-empty finite $\Gamma \cup \Delta$, $\Gamma ; \Delta \Vdash^{+}_{\mathcal{B}} \chi$ iff $ \forall \mathcal{C} (\mathcal{C} \supseteq \mathcal{B}) : \ $ if $ \ \Vdash_{\mathcal{C}}^{+} \phi$ for all $\phi \in \Gamma$ and $\Vdash_{\mathcal{C}}^{-} \psi$ for all $\psi \in \Delta$ then $\Vdash_{\mathcal{C}}^{+} \chi$.

\medskip

\item[\namedlabel{c:inf-}{(Inf $-$)}] For non-empty finite $\Gamma \cup \Delta$, $\Gamma ; \Delta \Vdash^{-}_{\mathcal{B}} \chi$ iff $ \forall\mathcal{C} (\mathcal{C} \supseteq \mathcal{B}) : \ $if$ \ \Vdash_{\mathcal{C}}^{+} \phi$ for all $\phi \in \Gamma$ and $\Vdash_{\mathcal{C}}^{-} \psi$ for all $\psi \in \Delta$ then $\Vdash_{\mathcal{C}}^{-} \chi$.

\medskip

\item[\namedlabel{c:2int+}{($\BPR +$)}] For finite $\Gamma \cup \Delta$, $\Gamma ; \Delta \Vdash^{+} \phi$ iff $\Gamma ; \Delta \Vdash^{+}_{\mathcal{B}} \phi$ for all $\mathcal{B}$;

\medskip

\item[\namedlabel{c:2int-}{($\BPR -$)}] For finite $\Gamma \cup \Delta$, $\Gamma ; \Delta \Vdash^{-} \phi$ iff $\Gamma ; \Delta \Vdash^{-}_{\mathcal{B}} \phi$ for all $\mathcal{B}$;
\end{enumerate}
\caption{$\BPR$ semantic clauses}\label{fig:sc}
\end{figure}

Unlike in the semantics defined in \cite{barrosonascimento2025bilateralbaseextensionsemantics}, the requirement of epistemic adequacy for bases (in this case specifically the combination of logical and epistemic consistency) allow us to prove that proofs and refutations indeed cannot coexist. This is not the case in $\Bint$ since, due to its paraconsistent nature, formulas are generally allowed to be provable and refutable at the same time.

The proof of the next result proceeds along the same lines as Lemma 2 in \cite{barrosonascimento2025bilateralbaseextensionsemantics}, adapted to a framework in which $\top$ and $\bot$ are regarded as atomic formulas.
\begin{lemma} If $\Gamma ; \Delta\Vdash^{*}_{\mathcal{B}} \phi$ and $\mathcal{B} \supseteq \mathcal{C}$ then $\Gamma ; \Delta \Vdash^{*}_{\mathcal{C}} \phi$.
    
\end{lemma}

The result below is proved in Appendix~\ref{app:sec4}. Notice that, although a similar model-theoretic result for a different language is stated (but not proved) in \cite{Wansing1995-HEISNI}, any proof would look significantly different from the one presented here.

\begin{theorem}\label{theorem:positivenegativesupport}
For any $\mathcal{B}$, $\Vdash^{+}_{\B} \phi$ implies $\nVdash^{-}_{\B} \phi$ and  $\Vdash^{-}_{\B} \phi$ implies $\nVdash^{+}_{\B} \phi$.
\end{theorem}

Just like the proof of Lemma~\ref{atomicpremises}, the proof of Theorem~\ref{theorem:positivenegativesupport} only works due the harmonic relationship existing between proof and refutation conditions. It also uses both the epistemic and logical consistency constraints in an essential way. If bases were not required to be epistemically consistent\footnote{In fact, if bases were only required to be logically consistent and unit complete we would obtain an alternative $\Bes$ for $\Bint$ instead of one for $\BPR$.} we could have a base with both $\Vdash^{+}_{\mathcal{B}} p$ and $\Vdash^{-}_{\mathcal{B}} p$, and if instead of using logical consistency we followed \cite{barrosonascimento2025bilateralbaseextensionsemantics} and used the traditional $\Bes$ definition of proofs of $\bot$ (as well as refutations of $\top$) in terms of atomic explosion ({\em e.g.} $\Vdash^{+}_{\mathcal{B}} \bot$ iff $\Vdash^{+}_{\mathcal{B}} p$ and $\Vdash^{-}_{\mathcal{B}} p$ for all atoms $p$) then in any base $\mathcal{B}$ with $\Vdash^{+}_{\mathcal{B}} p$ and $\Vdash^{-}_{\mathcal{B}} p$ for all $p \in \At$ we would also have $\Vdash^{+}_{\mathcal{B}} \phi$ and $\Vdash^{-}_{\mathcal{B}} \phi$ for all $\phi$. As such, Theorem~\ref{theorem:positivenegativesupport}, which shows that the semantics in fact codifies our pre-theoretic intuitions about $\BPR$, can only be obtained by this novel combination of the bilateralist ideas presented in \cite{barrosonascimento2025bilateralbaseextensionsemantics} with the treatment of consistency developed in \cite{EcumenicalPTS}.

\section{Soundness and Completeness}\label{sec:sc}

Having established normalisation and its immediate consequences, we now turn to the semantic adequacy of our system. The aim is to show that $\BPR$ is both sound and complete with respect to the $\Bes$ proposed in Section \ref{sec:semantics}. The proof of Soundness follows closely that of $\BintNs$ in~\cite{barrosonascimento2025bilateralbaseextensionsemantics}, differing only in that it has two additional cases corresponding to the instances of the $PR$ rule. The proof of completeness, however, requires a more detailed analysis, since our semantic framework presupposes that the bases are consistent. In this context, consistency will play the same crucial structural role as it plays in \cite{EcumenicalPTS}.

The soundness proof for $\BPR$ is given in Appendix~\ref{app:sec5}. Since the semantic treatment of $\bot$ and $\top$ here differs from the proof in~\cite{barrosonascimento2025bilateralbaseextensionsemantics} only in that we regard them as atomic formulas rather than connectives, no new cases arise in the auxiliary lemmas, which therefore have their proofs omitted.
%The only substantial additions concern the two instances of the co-ordination rule $PR$, whose soundness must be verified.

\begin{theorem}\label{thm:soundness}
    If $\Gamma; \Delta \vdash^{*}_{\BPR} \phi$ then $\Gamma; \Delta \Vdash^{*} \phi$.
\end{theorem}

In contrast to soundness, the completeness proof is more involved. The overall strategy follows that of \cite{barrosonascimento2025bilateralbaseextensionsemantics}, based on Sandqvist's simulation base construction. However, in our case we must additionally ensure that the simulation bases used in the argument satisfy the constraints required in the definition of $\Bes$.
To achieve this, we show that the simulation bases are themselves consistent. The key step is to establish atomic normalisation for these bases. Since normalisation has already been proved for $\BPR$, this step follows directly. The resulting normalisation procedure then allows us to extend the consistency argument to the simulation bases, ensuring that they meet the requirements of our semantic framework. 

\begin{definition}
   If $\Gamma$ is a set of formulas, then $\Gamma^{Sub} = \{ \phi \ | \ \phi $ is a subformula of $\psi$ and $\psi \in \Gamma\}$.
\end{definition}

\begin{definition}
  An {\em atomic mapping} $\alpha$ for a set $\Gamma^{Sub}$ is any injective function $\alpha: \Gamma^{Sub} \mapsto \At$ such that, for all $p \in \Gamma^{Sub}$ with $p \in \At$, $\alpha(p) = p$.

\end{definition}

  For every formula $\phi\in\Gamma^{Sub}$, we will denote the atom $\alpha(\phi)$ by $p^{\phi}$. Note that, since $\alpha$ is injective, this notation is well defined, as no two different formulas may be assigned to the same atom.
\medskip

  Given an atomic mapping $\alpha$ for $\Gamma$, a {\em simulation base} $\mathcal{U}$ is a base which contains rules mimicking the rules of system $\BPR$. The construction of this base is exactly as in \cite{barrosonascimento2025bilateralbaseextensionsemantics}, with the addition of the following rules for all $\phi \in \Gamma$:
  \vspace{-0.7cm}
\begin{prooftree}
    \AxiomC{}
    \UnaryInfC{$p^{\phi}$}
    \AxiomC{}
    \doubleLine
    \UnaryInfC{$p^{\phi}$}
    \RightLabel{\tiny $PR(+),p$}

    \BinaryInfC{$q$}
    \DisplayProof
    \quad \quad \quad
    \AxiomC{}
    \UnaryInfC{$p^{\phi}$}
    \AxiomC{}
    \doubleLine
    \UnaryInfC{$p^{\phi}$}
    \doubleLine
    \RightLabel{\tiny $PR(-),p$}
    \BinaryInfC{$q$}
\end{prooftree}
 \vspace{-0.2cm}
Furthermore, notice that instances of $p^{\bot}$ and $p^{\top}$ are now written as simply $\bot$ and $\top$, as in the present framework those are treated as if they were atoms. 

%  Note that the subformula property is crucial in the definition of the simulation base, since the atomic mapping is defined only for atoms occurring as subformulas of the relevant formulas. The subformula property ensures that the atomic mapping is well-defined and that every step in a derivation in $\BPR$ can be simulated by an application of some rule in $\mathcal{U}$.

%\Victor{This is not entirely the case. I think that if a logic does not satisfy the subformula property then in general we cannot give it a PTS because validity cannot be inductively defined by considering only subformulas, but we don't need to explicitly rely on the subformula property (as in: use it as a lemma) in order to build simulation bases or prove soundness/completeness.}
%\Maria{But if a random (non-subformula of premises of conclusion) non-atomic formula appears in a derivation, what would be its counterpart when mirroring the derivation with the rules in the simulation base? }
  
The proof of atomic normalisation follows standard lines. The atomic reductions are analogous to those used in our normalisation proof. The following example illustrates how these reductions can be applied in the setting of simulation bases.

\begin{example}
The following is a proper reduction for simulation bases:
\vspace{-0.3cm}
\begin{prooftree}
    \AxiomC{$[p^{\phi}]$}
    \noLine
    \UnaryInfC{$\Sigma$}
    \UnaryInfC{$p^{\psi}$}             \RightLabel{\scriptsize{$\rightarrow I(+), p^{\phi \to \psi}$}}
    \UnaryInfC{$p^{\phi \rightarrow \psi}$}
    \AxiomC{$\Sigma'$}
    \UnaryInfC{$p^{\phi}$}    \RightLabel{\scriptsize{$\rightarrow E(+), p^{\phi \to \psi}$}}
    \BinaryInfC{$p^{\psi}$} 
    \DisplayProof
    \quad
    $\Longrightarrow$    
    \quad
    \AxiomC{$\Sigma'$}
    \UnaryInfC{$p^{\phi}$}
    \noLine
    \UnaryInfC{$\Sigma$}
    \UnaryInfC{$p^{\psi}$} 
\end{prooftree}
\vspace{-0.3cm}
\end{example}
\begin{theorem}\label{thm:atomicnorm}
    Let $\Pi$ be a derivation of $\phi$ from the set of assumptions $\Gamma$ and contra-assumptions $\Delta$, in $\mathcal{U}$. Then $\Pi$ reduces to a normal derivation $\Pi'$ of $\phi$ from the set of assumptions $\Gamma$ and contra-assumptions $\Delta$.
\end{theorem}

\begin{corollary}\label{lemma:simulationconsistent}
    $\mathcal{U}$ is consistent.
\end{corollary}
The proofs of these results can be obtained through a straightforward adaptation of the proofs of Theorem~\ref{thm:normalisation} and Corollary~\ref{cor:consistencyBPR} to the atomic setting.

\begin{lemma}\label{lemma:atomicsupportiffderivability}
Let $\phi \in \Gamma^{Sub}$, $\alpha$ be an atomic mapping for $\Gamma^{Sub}$ and $\mathcal{U}$ a simulation base defined with $\alpha$. Then for all epistemically adequate $\mathcal{B} \supseteq \mathcal{U}$, it holds that $\Vdash^{*}_{\mathcal{B}} \phi$ iff $\vdash^{*}_{\mathcal{B}} p^{\phi}$.
\end{lemma}

Due to our use of the consistency constraint, our completeness proof has to use the strategy developed in \cite{EcumenicalPTS} instead of the traditional strategy for $\Bes$. The proof can be found in Appendix~\ref{app:sec5}. 

\begin{theorem} \label{thm:completeness}
    If $\Gamma;\Delta\Vdash^* \phi$ then $\Gamma;\Delta\vdash^*_{\BPR} \phi$.
\end{theorem}

We can prove now that weak falsity does not imply strong falsity in $\BPR$.
%because it is not the case that from a proof showing the impossibility of proving $\phi$ we can extract an actual counterexample to $\phi$, the same once again holding for an analogous concept of weak provability. That this is in fact the case in our new logic can be shown  by two simple counterexamples. 
In fact, consider a base $\mathcal{B}$ containing only the rule $R1$ and a base $\mathcal{C}$ containing only the rule $R2$ below:
\vspace{-0.4cm}
\begin{prooftree}
\AxiomC{}
    \UnaryInfC{$p$}
    \RightLabel{\scriptsize{$R1$}}
    \UnaryInfC{$\bot$}
    \DisplayProof
    \qquad
    \qquad
        \AxiomC{}
        \doubleLine
        \UnaryInfC{$p$}
        \doubleLine
    \RightLabel{\scriptsize{$R2$}}
    \UnaryInfC{$\top$}
\end{prooftree}
\vspace{-0.3cm}
Since our bases satisfy the consistency constraint there cannot be any $\mathcal{D} \supseteq \mathcal{B}$ with $\Vdash^{+}_{\mathcal{D}} p$ or any  $\mathcal{E} \supseteq \mathcal{C}$ with $\Vdash^{-}_{\mathcal{E}} p$ since, if there were, we would obtain deductions showing $\vdash^{+}_{\mathcal{D}} p$ and $\vdash^{-}_{\mathcal{E}} p$ which could be composed with the rules to show that such bases are inconsistent. So from the semantic clauses $(\to +)$ and $(\mapsfrom -)$ we conclude $\Vdash^{+}_{\mathcal{D}} p \to \bot$ and $\Vdash^{-}_{\mathcal{E}} \top \mapsfrom p$. On the other hand, since there are no refutation rules in $\mathcal{B}$ and no proof rules in $\mathcal{C}$ we immediately conclude $\nvdash^{-}_{\mathcal{B}} p$ and $\nvdash^{+}_{\mathcal{C}} p$, meaning that in $\mathcal{B}$ we can prove weak falsity without providing a constructive counterexample and in $\mathcal{C}$ we can prove weak provability without providing a constructive proof, showing that the concepts interact as expected.

\section{Concluding remarks and future work}\label{sec:conc}
In this paper, we develop a form of bilateralism in which proofs and refutations are distinct but incompatible assertoric acts, and we establish syntactic and semantic results for the logic $\BPR$.

On the syntactic side, we show that our natural deduction system--extending an adapted version of Wansing's $\Bint$--is weakly normalising. As corollaries, we obtain the subformula property and consistency, ensuring that all proofs and refutations can be reduced to analytic form.

On the semantic side, we introduce a base-extension semantics that is sound and complete with respect to the proof system. Because it relies on the explicit construction of proofs and counterexamples via atomic rules, the semantics remains faithful to the mathematical practice of proving and refuting, while also supporting epistemic interpretations.

Finally, we show that $\BPR$ internalises constructive falsity via refutation and supports the definition of further operators, including intuitionistic negation. Its relationship to Nelson's $\Nt$ mirrors that between $\Bint$ and $\Nf$.

There are several directions for future work. A natural next step is to develop alternative proof systems for $\BPR$, in particular a sequent calculus along the lines of~\cite{Ayhan-seq}, enabling systematic comparison with existing frameworks such as~\cite{negri-sn}. This also opens the door for automatic reasoning. Further directions include first-order extensions, applications to other logics with strong negation, and a deeper analysis of the proof-theoretic and semantic structure of $\BPR$.

\newpage

\bibliography{references}

\newpage

\appendix
%!TEX root = main.tex
% !TEX spellcheck = en-UK

%!TEX root = main.tex
% !TEX spellcheck = en-UK

\section{Extended Discussion of Related Work}\label{sec:comp}
{\bf On $\Nf$, $\Nfb$ and $\Bint$.}
Whereas $\Nfb$ is an extension of Nelson's $\Nf$ which allows definition of a weak negation, Wansing's $\Bint$ is instead both an extension of intuitionistic and dual-intuitionistic logic. This is interesting to point out because, as mentioned, Nelson's logics were originally proposed as alternatives to intuitionistic logic, but the similarities between $\Nfb$ and $\Bint$ suggests that both logics are not very far apart. In fact, a recent result \cite{oddsson2025strongnegationdefinable2int} shows that strong negation $\nn\phi$ is definable in $\Bint$ as the formula $(\phi \land (\phi \to (\phi \mapsfrom \phi))) \lor ((\phi \to \phi) \mapsfrom \phi)$, dispelling any doubts we might have concerning the fact that adding strong negation to $\Bint$ does not make it more expressive. This is particularly interesting for a number of reasons. First, proofs of strong negation are equivalent to refutations and refutations of strong negation are equivalent to proofs, in the sense that from the proof in \cite{oddsson2025strongnegationdefinable2int} plus soundness and completeness it follows that ($\Vdash^{+}_{\mathcal{B}} \nn\phi$ iff $\Vdash^{-}_{\mathcal{B}} \phi$) and ($\Vdash^{-}_{\mathcal{B}} \nn\phi$ iff $\Vdash^{+}_{\mathcal{B}} \phi$). This allows Nelson's negation to function as a sort of {\em structural} bilateral operator which allows direct passage from proofs to refutations. Second, from the definability of $\nn\phi$ in $\Bint$ it follows that the semantic harmony results shown for arbitrary formulas in \cite{barrosonascimento2025bilateralbaseextensionsemantics} also applies to strong negation\footnote{This result could also be proved directly by modifying Definition 25 in \cite{barrosonascimento2025bilateralbaseextensionsemantics} by including the clause $(\sim \phi)^{\mathbb{D}} = \sim (\phi)^{\mathbb{D}}$. }, which in turn allows its use for the obtainment of semantic embeddings. Third, from the fact that proofs of $\nn\phi$ are equivalent to refutations of $\phi$ it follows that addition of the strong negation operator to intuitionistic logic is equivalent to the addition of $\Bint$'s concept of refutation to it, as would also be the case of any similar ``strong affirmation" operator added to dual-intuitionistic logic.

\medskip
\noindent
{\bf On $\BPR$ and $\Nt$.}
We argue that $\BPR$ differs crucially from $\Nt$ inasmuch Nelson's strong negation $\sim$ is involutive by definition: a P-realizer of $\sim \phi$ is exactly an N-realizer of $\phi$, and conversely. This collapses proofs of $\sim \phi$ into refutations of $\phi$ and refutations of $\sim \phi$ into proofs of $\phi$. Such duality is essential to $\Nt$, but it does not reflect the epistemic notion of constructive falsity we aim to model. In mathematics one does not ``refute a refutation'' as refutation is an act of producing a counterexample, not an operator closed under involution. On the other hand, $\BPR$ captures this asymmetry directly. Proofs and refutations are primitive epistemic acts, governed by coordination rules that enforce incompatibility without relying on an involutive negation. This allows our system to represent constructive reasoning in terms of constructions and counterexamples, and to support a semantics where proofs and refutations cannot coexist. For these reasons, $\BPR$ provides a more accurate model of constructive mathematical practice than $\Nt$.

%\cyan{The non-involutive character of refutations does not preclude a conceptually sound definition of strong negation. We justify this by making a distinction between the \textit{act of refuting} and the \textit{claim that a refutation was performed}. A refutation of $\phi$ is an assertoric act consisting in the presentation of a counterexample to $\phi$, whereas an assertion of $\sim \phi$ consists in the affirmation that such a counterexample has been provided. In Rumfitt's terms, refuting something is like answering ``No" to a question, while asserting a strong negation is like saying ``Their answer was no". Once the distinction is made, it is not hard to see that involutiveness should only apply to the latter notion: an act cannot be refuted, but the claim that an act has been performed can. In other words, strong negation is a proposition whose semantic content is a non-propositional act. In a mathematical setting, to refute the claim that a counterexample to $\phi$ has been provided is to provide a proof of $\phi$, and to be in a position to assert that $\phi$ has a counterexample is to be in possession of said counterexample, but it does not make sense to ask how to refute the act of writing a counterexample.}

\medskip
\noindent
{\bf On $\BPR$ and other bilateral logics.}
Although it is conceptually closer to $\Bint$ and $\Nt$, there are important comparisons to be made between $\BPR$ and other bilateral logics\footnote{The logics considered in this paragraph use positive and negative labels for formulas instead of proof and dual proof rules, but that is an inessential change.}. For instance, $\BPR$ can be obtained from the logic $\mathbf{HB_{2}}$~\cite{Valle-Inclan2023-VALHAN} by adding rules for the units $\top$ and $\bot$ and for co-implication, removing the classical co-ordination rule and using the traditional rules for assertion of disjunction and refutation of conjunction instead of the modified ones\footnote{The modified rules define such connectives in terms of the \textit{disjunctive syllogism}, essentially defining $\phi \lor \psi$ as $\neg \phi \to \psi$. Since $\phi \lor \phi$ is not equivalent to $\neg \phi \to \psi$ in intuitionistic logic, this modification could not be sensibly used in $\BPR$.}. Since $\mathbf{HB_{2}}$ is equivalent to Rumfitt's calculus [pg. 384]\cite{Valle-Inclan2023-VALHAN}, $\BPR$ can also be interpreted as a intuitionistic version of Rumfitt's logic. $\mathbf{HB_2}$ is itself a variant of the logic $\mathbf{HB_{1}}$ presented in \cite{Kurbis2016-KRBSCO}, which is also equivalent to Rumfitt's logic. Those three logics have in common, as do the slightly less similar $C-intelim$ logic in \cite{AgostinoContamination}, the fact that they aim to provide proof-theoretic characterisations of \textit{classical logic}, a goal we do not share. They also define $\bot$ as a punctuation mark instead of a formula, which prevents them from being able to draw comparisons between weak and strong negation. The difference in the formulation of the rules and the fact that $\BPR$ is closer to $\Bint$ and $\Nt$ is mostly explained by technical implications of this change in goal.

\medskip
\noindent
{\bf On $\BPR$ and Rumfitt.}
The new rule $PR$ is the natural deduction equivalent of one of Ian Rumfitt's bilateral \textit{ co-ordination principles}, namely his \textit{Law of Non-Contradiction} \cite[pg. 804]{Rumfitt2000-RUMYAN-2}. The only noticeable differences are that Rumfitt's principle coordinates only atoms, whereas ours co-ordinate arbitrary formulas, and our rule can conclude an arbitrary formula, whereas Rumfitt's principle always has $\bot$ as its conclusion. The first modification is optional, but the second is essential for obtainment of the subformula principle. $\BPR$ does not contain his \textit{Reductio} principle, which would allow one to obtain a refutation of $\phi$ when from a proof of $\phi$ follows one of $\bot$, as well as a proof of $\phi$ when from a refutation of $\phi$ follows a proof of $\bot$. The \textit{Non-contradiction} principle is at first defined by Rumfitt only for atomic premises, but he uses the \textit{Reductio} to extend it to arbitrary formulas. We do the opposite by proving that all deductions can be reduced to ones in which all instances of $PR$ have atomic conclusions, a proof that essentially relies on the harmony between proofs and refutations characteristic of bilateral systems.

%That such a result can be proved without the \textit{Reductio} in $\BPR$ is relevant in its own right, as it shows that \textit{Non-contradiction} and \textit{Reductio} are independent under certain conditions and elucidates why those do not obtain in Rumfitt's logic.}

Both our acceptance of \textit{Non-contradiction} and our rejection of \textit{Reductio} can be justified on conceptual grounds. Inclusion of $PR$ is warranted by the fact that we specifically want to formalise epistemic entities for which such principle holds, as is the case of mathematical proofs. Exclusion of the \textit{Reductio} is justified by the fact that this would allow rejections to be established indirectly, thus collapsing weak and strong falsity.  This is not a problem for Rumfitt, as he aims to provide a proof-theoretic justification of classical logic\footnote{Whether he succeeds or not is a contentious matter. See, for instance, Dummett's rebuttal in \cite{DummettYesNo}. The original proposal also has technical problems, as pointed out in \cite{FernandoFerreiraBilateral}, although Rumfitt argues that this poses no threat to bilateralist accounts \cite{RumfittAnswertoFerreira}.}, but since this is precisely the distinction we wish to draw the principle cannot be accepted here. This would be akin to claiming that, by producing a contradiction from a claim, we could always obtain a particular counterexample for it, which is a untenable position.

\section{Proofs of some selected results} \label{sec:appendix}
%!TEX root = main.tex
% !TEX spellcheck = en-UK

\subsection{Proofs from Section~\ref{sec:pt}}\label{app:sec3}
\begin{definition}
    Two formula occurrences are {\em side connected} if they are premises of the same rule application.
\end{definition}

\begin{proof}[Proof of Lemma~\ref{atomicpremises}]
    We prove the result by induction on the complexity value $\langle d, n \rangle$ of $\Pi$, showing that, if %$\langle d,  n \rangle$ and 
it is not the case that $d = n = 0$, then $\Pi$ reduces to a deduction $\Pi'$ which has complexity value $\langle d', n' \rangle$ for $\langle d', n' \rangle <\langle d, n\rangle$ and the same premises and conclusion. Clearly, if $d = n = 0$ then by the definition of complexity values we have that all instances of $PR$, $\bot (+)$ and $\top (-)$ have atomic premises and conclusion.

First take a formula occurrence $\phi$ in $\Pi$ which is a premise or conclusion with degree $d$ of an instance of $PR$, $\bot (+)$ or $\top (-)$\footnote{Clearly, in the case of $\bot (+)$ and $\top (-)$ it can only be the rule's conclusion.}, satisfying the requirement that no instance of $PR$, $\bot (+)$ or $\top (-)$ above the formula occurrence has premises or conclusions with degree $d$ (if $\phi$ is the conclusion of $PR$ this includes the rule application of which the formula is a conclusion). If the formula is a premise of $PR$ we also require that no such instance occurs above the formula occurrence side connected with it\footnote{As usual, we can find such a formula occurrence by picking any premise or conclusion with degree $d$ of an application of $PR$, $\bot (+)$ or $\top (-)$ and, if there is some rule application above it or above its side connected formula which makes it violate the requirements, we just switch to that upper occurrence. Since deductions are finite, by repeating this procedure we eventually reach a formula occurrence satisfying the desired conditions. Notice that, in the case of applications of $PR$ with premises of the same degree, the effect of this procedure is that we have to deal with the premises of the application before dealing with the conclusion.}. We show the reductions for when the application of $PR$ contains a single line, since the cases where it contains a double line instead are always analogous. The reductions are defined based on the shape of $\phi$:

\begin{itemize}
    \item[($\phi = \psi \land \chi$).]  We rewrite the deduction as follows:

\medskip

    \begin{bprooftree}
  \hspace{-50pt}
        \AxiomC{$\Sigma$}
        \UnaryInfC{$\psi \land \chi$}
        \AxiomC{$\Sigma'$}
        \doubleLine
        \UnaryInfC{$\psi \land \chi$}
        \RightLabel{\scriptsize{$PR$}}
        \BinaryInfC{$\sigma$}
        \noLine
        \UnaryInfC{$\Sigma''$}
\end{bprooftree}
\quad
$\Longrightarrow$
\quad
\begin{bprooftree}
  \AxiomC{$\Sigma'$}
        \doubleLine
        \UnaryInfC{$\psi \land \chi$}
          \AxiomC{}
          \doubleLine
          \RightLabel{\scriptsize{1}}
          \UnaryInfC{$\psi$}
            \AxiomC{$\Sigma$}
        \UnaryInfC{$\psi \land \chi$}
        \RightLabel{\scriptsize{$E_{1} \land (+)$}}
        \UnaryInfC{$\psi$}
             \RightLabel{\scriptsize{$PR$}}
        \BinaryInfC{$\sigma$}
         \AxiomC{}
          \RightLabel{\scriptsize{2}}
          \doubleLine
          \UnaryInfC{$\chi$}
            \AxiomC{$\Sigma$}
        \UnaryInfC{$\psi \land \chi$}
               \RightLabel{\scriptsize{$E_{2} \land (+)$}}
        \UnaryInfC{$\chi$}
         \RightLabel{\scriptsize{$PR$}}
        \BinaryInfC{$\sigma$}
        \RightLabel{\scriptsize{$E_{1} \land (-),  1, 2$}}
        \TrinaryInfC{$\sigma$}
        \noLine
        \UnaryInfC{$\Sigma ''$}
\end{bprooftree}

\medskip

%\cyan{Mental note by Victor: remove ambiguity in definition of degree by stating that if a formula occurrence appears both as premise and as conclusion of such a rule it is counted twice in the degree.}

%\cyan{Mental note by Victor (for extended version of the paper): we can use an alternative proof of the lemma for $\bot$, but this changes the reductions! And the shape of the reduction is different if both premises of PR are maximal and if only one of them is!}

Notice that this reduction replaces the original application of $PR$ by two others and creates a copy of $\sigma$, which might have the same degree as $\psi \land \chi$. This is not problematic because, in this case, the new deduction contains two formula occurrences with degree $d$ appearing as the premise or conclusion of $PR$ (the conclusions with shape $\sigma$), whereas the original deduction had three (the two premises with shape $\psi \land \chi$ and the conclusion $\sigma$), so the complexity value of the deduction still decreases by 1. Notice also that, if the formula occurrence $\sigma$ in the original derivation is also the premise of an application of $PR$, then it has to be counted twice in the inductive measure, as otherwise the new deduction would have the same inductive value as the original.

Since all applications of $PR$, $\bot (+)$ and $\top (-)$ in $\Sigma$ above the original application of $PR$ which are duplicated by this procedure have premises with degree less than $d$ and the two new applications of $PR$ have premises with degree lesser than $d$, either the occurrences of $\psi \land \chi$ on $\Pi$ were the last ones appearing on $\Pi$ as premises or conclusions of $PR$, $\bot (+)$ or $\top (-)$ which have degree $d$, in which case $d' < d$, or there are other such applications, in which case $d' = d$ but $n' < n$, thus we have $\langle d', n' \rangle < \langle d, n \rangle$. The only difference in the reduction if we were considering an application of $PR$ with a double line is that the new deduction would instead have two occurrences of $PR$ concluding refutations of $\sigma$ and an instance of $E \land (-)$ with double lines.

\bigskip

    \item[($\phi = \psi \lor \chi$).] We rewrite the deduction as follows:

    \begin{bprooftree}
    \hspace{-50pt}
        \AxiomC{$\Sigma$}
        \UnaryInfC{$\psi \lor \chi$}
        \AxiomC{$\Sigma'$}
        \doubleLine
        \UnaryInfC{$\psi \lor \chi$}
                \RightLabel{\scriptsize{$PR$}}
        \BinaryInfC{$\sigma$}
        \noLine
        \UnaryInfC{$\Sigma''$}
    \end{bprooftree}
\quad
$\Longrightarrow$
\quad
\begin{bprooftree}
         \AxiomC{$\Sigma$}
        \UnaryInfC{$\psi \lor \chi$}
          \AxiomC{}
          \RightLabel{\scriptsize{1}}
          \UnaryInfC{$\psi$}
            \AxiomC{$\Sigma'$}
            \doubleLine
        \UnaryInfC{$\psi \lor \chi$}
        \doubleLine
        \RightLabel{\scriptsize{$E_{1} \lor (-)$}}
        \UnaryInfC{$\psi$}
             \RightLabel{\scriptsize{$PR$}}
        \BinaryInfC{$\sigma$}
         \AxiomC{}
          \RightLabel{\scriptsize{2}}
          \UnaryInfC{$\chi$}
            \AxiomC{$\Sigma'$}
            \doubleLine
        \UnaryInfC{$\psi \lor \chi$}
        \doubleLine
               \RightLabel{\scriptsize{$E_{2} \lor (-)$}}
        \UnaryInfC{$\chi$}
         \RightLabel{\scriptsize{$PR$}}
        \BinaryInfC{$\sigma$}
        \RightLabel{\scriptsize{$E_{1} \lor (+),  1, 2$}}
        \TrinaryInfC{$\sigma$}
        \noLine
        \UnaryInfC{$\Sigma ''$}
\end{bprooftree}
\medskip

The same reasoning used in the case for $\phi = \psi \land \chi$ shows $\langle d', d' \rangle < \langle d, n \rangle$ for the complexity degree $\langle d', n' \rangle$ of $\Pi'$. 

\bigskip

    \item[($\sigma = \psi \to \chi$).]  We rewrite the deduction as follows:
    
    \begin{bprooftree}
        \AxiomC{$\Sigma$}
        \UnaryInfC{$\psi \to \chi$}
        \AxiomC{$\Sigma'$}
        \doubleLine
        \UnaryInfC{$\psi \to \chi$}
                \RightLabel{\scriptsize{$PR$}}
        \BinaryInfC{$\sigma$}
        \noLine
        \UnaryInfC{$\Sigma''$}
    \end{bprooftree}
\quad
$\Longrightarrow$
\quad
\begin{bprooftree}
         \AxiomC{$\Sigma$}
        \UnaryInfC{$\psi \to \chi$}
        \AxiomC{$\Sigma'$}
     \doubleLine
        \UnaryInfC{$\psi \to \chi$}
        \RightLabel{\scriptsize{$E_{1} \to (-)$}}
        \UnaryInfC{$\psi$}
         \RightLabel{\scriptsize{$E \to (+)$}}
        \BinaryInfC{$\chi$}
            \AxiomC{$\Sigma'$}
     \doubleLine
        \UnaryInfC{$\psi \to \chi$}
        \doubleLine
        \RightLabel{\scriptsize{$E_{2} \to (-)$}}
        \UnaryInfC{$\chi$}
    \RightLabel{\scriptsize{$PR$}}
        \BinaryInfC{$\sigma$}
        \noLine
        \UnaryInfC{$\Sigma''$}
\end{bprooftree}

\bigskip

\item[($\phi = \psi \mapsfrom \chi$)]  We rewrite the deduction as follows:

\bigskip

    \begin{bprooftree}
        \AxiomC{$\Sigma$}
        \UnaryInfC{$\psi \mapsfrom \chi$}
        \AxiomC{$\Sigma'$}
        \doubleLine
        \UnaryInfC{$\psi \mapsfrom \chi$}
                \RightLabel{\scriptsize{$PR$}}
        \BinaryInfC{$\sigma$}
        \noLine
        \UnaryInfC{$\Sigma''$}
    \end{bprooftree}
\quad
$\Longrightarrow$
\quad
\begin{bprooftree}
\AxiomC{$\Sigma'$}
         \doubleLine
        \UnaryInfC{$\psi \mapsfrom \chi$}
    \AxiomC{$\Sigma$}
        \UnaryInfC{$\psi \mapsfrom \chi$}
    \RightLabel{\scriptsize{$E_{2} \mapsfrom (+)$}}
        \doubleLine
        \UnaryInfC{$\chi$}
         \RightLabel{\scriptsize{$E \mapsfrom (-)$}}
         \doubleLine
        \BinaryInfC{$\psi$}
            \AxiomC{$\Sigma$}
        \UnaryInfC{$\psi \mapsfrom \chi$}
    \RightLabel{\scriptsize{$E_{1} \to (+)$}}
        \UnaryInfC{$\psi$}
    \RightLabel{\scriptsize{$PR$}}
        \BinaryInfC{$\sigma$}
        \noLine
        \UnaryInfC{$\Sigma''$}
\end{bprooftree}

\end{itemize}
\bigskip

If the formula occurrence is a conclusion of $PR$, $\bot (+)$ or $\top (-)$, the reductions to be applied are uniform. We exemplify them for $PR$ with a single line:

\bigskip

\begin{prooftree}
    \AxiomC{$\Sigma$}
    \UnaryInfC{$\chi$}
    \AxiomC{$\Sigma'$}
    \doubleLine
    \UnaryInfC{$\chi$}
    \RightLabel{\scriptsize{$PR$}}
    \BinaryInfC{$\phi \land \psi$}
    \noLine
    \UnaryInfC{$\Sigma''$}
    \DisplayProof
\quad
$\Longrightarrow$
\quad
    \AxiomC{$\Sigma$}
    \UnaryInfC{$\chi$}
    \AxiomC{$\Sigma'$}
    \doubleLine
    \UnaryInfC{$\chi$}
    \RightLabel{\scriptsize{$PR$}}
    \BinaryInfC{$\phi$}
        \AxiomC{$\Sigma$}
    \UnaryInfC{$\chi$}
    \AxiomC{$\Sigma'$}
    \doubleLine
    \UnaryInfC{$\chi$}
    \RightLabel{\scriptsize{$PR$}}
    \BinaryInfC{$\psi$}
    \RightLabel{\scriptsize{$I \land (+)$}}
    \BinaryInfC{$\phi \land \psi$}
    \noLine
    \UnaryInfC{$\Sigma''$}
\end{prooftree}

\begin{prooftree}
    \AxiomC{$\Sigma$}
    \UnaryInfC{$\chi$}
    \AxiomC{$\Sigma'$}
    \doubleLine
    \UnaryInfC{$\chi$}
    \RightLabel{\scriptsize{$PR$}}
    \BinaryInfC{$\phi \lor \psi$}
    \noLine
    \UnaryInfC{$\Sigma''$}
    \DisplayProof
\quad
$\Longrightarrow$
\quad
    \AxiomC{$\Sigma$}
    \UnaryInfC{$\chi$}
    \AxiomC{$\Sigma'$}
    \doubleLine
    \UnaryInfC{$\chi$}
    \RightLabel{\scriptsize{$PR$}}
    \BinaryInfC{$\phi$}
       \RightLabel{\scriptsize{$I_{1} \lor  (+)$}}
      \UnaryInfC{$\phi \lor \psi$}
      \noLine
      \UnaryInfC{$\Sigma ''$}
\end{prooftree}

\begin{prooftree}
    \AxiomC{$\Sigma$}
    \UnaryInfC{$\chi$}
    \AxiomC{$\Sigma'$}
    \doubleLine
    \UnaryInfC{$\chi$}
    \RightLabel{\scriptsize{$PR$}}
    \BinaryInfC{$\phi \to \psi$}
    \noLine
    \UnaryInfC{$\Sigma''$}
    \DisplayProof
\quad
$\Longrightarrow$
\quad
    \AxiomC{$\Sigma$}
    \UnaryInfC{$\chi$}
    \AxiomC{$\Sigma'$}
    \doubleLine
    \UnaryInfC{$\chi$}
    \RightLabel{\scriptsize{$PR$}}
    \BinaryInfC{$\psi$}
    \RightLabel{\scriptsize{$I \to (+)$}}
      \UnaryInfC{$\phi \to \psi$}
      \noLine
      \UnaryInfC{$\Sigma ''$}
\end{prooftree}

\begin{prooftree}
    \AxiomC{$\Sigma$}
    \UnaryInfC{$\chi$}
    \AxiomC{$\Sigma'$}
    \doubleLine
    \UnaryInfC{$\chi$}
    \RightLabel{\scriptsize{$PR$}}
    \BinaryInfC{$\phi \mapsfrom \psi$}
    \noLine
    \UnaryInfC{$\Sigma''$}
    \DisplayProof
\quad
$\Longrightarrow$
\quad
    \AxiomC{$\Sigma$}
    \UnaryInfC{$\chi$}
    \AxiomC{$\Sigma'$}
    \doubleLine
    \UnaryInfC{$\chi$}
    \RightLabel{\scriptsize{$PR$}}
    \BinaryInfC{$\phi$}
        \AxiomC{$\Sigma$}
    \UnaryInfC{$\chi$}
    \AxiomC{$\Sigma'$}
    \doubleLine
    \UnaryInfC{$\chi$}
    \RightLabel{\scriptsize{$PR$}}
    \doubleLine
    \BinaryInfC{$\psi$}
    \RightLabel{\scriptsize{$I \mapsfrom (+)$}}
    \BinaryInfC{$\phi \mapsfrom \psi$}
    \noLine
    \UnaryInfC{$\Sigma''$}
\end{prooftree}

The restrictions on the formula occurrence we pick guarantee that no formula occurrence which is a premise or conclusion of one of the relevant rules occur in $\Sigma$ ir $\Sigma'$, as well as that $\chi$ has degree lesser than that of the selected formula, so in any case either there are no other  relevant formula occurrences with the same degree in the deduction, so it reduces to a $\Pi'$ with $d < d'$, or there are other occurrences and it reduces to a $\Pi'$ with $d = d'$ but $n' < n$, so in any case we obtain a deduction with complexity degree $\langle d', n' \rangle$ such that $\langle d', n' \rangle < \langle d, n \rangle$. The cases for $\bot (+)$ are analogous, and the cases for $(\phi \land \psi)$, $(\phi \lor \psi)$, $(\phi \to \psi)$ and $(\phi \mapsfrom \psi)$ occurring as a consequence of $\top (-)$ are respectively analogous to those of $(\phi \lor \psi)$,  $(\phi \land \psi)$,  $(\phi \mapsfrom \psi)$ and  $(\phi \to \psi)$.
\end{proof}

\begin{proof}[Proof of Lemma~\ref{MariaLemma}] We prove the result by induction of the number of formula occurrences which are at the same time conclusion of an application of $PR$, $\bot (+)$ or $\top (-)$. The proof is quite simple, as we just delete rule applications and possibly substitute applications of $PR$ with single lines by ones with double lines (and vice-versa). We exemplify by stating the cases in which the first rule application is of shape $PR$, the remaining cases being analogous.

\begin{prooftree}
    \AxiomC{$\Sigma$}
    \UnaryInfC{$\phi$}
    \AxiomC{$\Sigma'$}
    \doubleLine
    \UnaryInfC{$\phi$}
    \RightLabel{\scriptsize{$PR$}}
    \BinaryInfC{$\psi$}
    \AxiomC{$\Sigma''$}
    \doubleLine
    \UnaryInfC{$\psi$}
      \RightLabel{\scriptsize{$PR$}}
    \BinaryInfC{$\chi$}
    \noLine
    \UnaryInfC{$\Sigma'''$}
    \DisplayProof
    \quad
    $\Longrightarrow$
     \AxiomC{$\Sigma$}
    \UnaryInfC{$\phi$}
    \AxiomC{$\Sigma'$}
    \doubleLine
    \UnaryInfC{$\phi$}
    \RightLabel{\scriptsize{$PR$}}
    \BinaryInfC{$\chi$}
    \noLine
    \UnaryInfC{$\Sigma'''$}
\end{prooftree}

\begin{prooftree}
    \AxiomC{$\Sigma$}
    \UnaryInfC{$\phi$}
    \AxiomC{$\Sigma'$}
    \doubleLine
    \UnaryInfC{$\phi$}
    \RightLabel{\scriptsize{$PR$}}
    \BinaryInfC{$\psi$}
    \AxiomC{$\Sigma''$}
    \doubleLine
    \UnaryInfC{$\psi$}
      \RightLabel{\scriptsize{$PR$}}
      \doubleLine
    \BinaryInfC{$\chi$}
    \noLine
    \UnaryInfC{$\Sigma'''$}
    \DisplayProof
    \quad
    $\Longrightarrow$
     \AxiomC{$\Sigma$}
    \UnaryInfC{$\phi$}
    \AxiomC{$\Sigma'$}
    \doubleLine
    \UnaryInfC{$\phi$}
    \RightLabel{\scriptsize{$PR$}}
    \doubleLine
    \BinaryInfC{$\chi$}
    \noLine
    \UnaryInfC{$\Sigma'''$}
\end{prooftree}

\bigskip

\begin{prooftree}
    \AxiomC{$\Sigma$}
    \UnaryInfC{$\phi$}
    \AxiomC{$\Sigma'$}
    \doubleLine
    \UnaryInfC{$\phi$}
    \RightLabel{\scriptsize{$PR$}}
    \BinaryInfC{$\bot$}
      \RightLabel{\scriptsize{$\bot (+)$}}
    \UnaryInfC{$\chi$}
    \noLine
    \UnaryInfC{$\Sigma'''$}
    \DisplayProof
    \quad
    $\Longrightarrow$
   \AxiomC{$\Sigma$}
    \UnaryInfC{$\phi$}
    \AxiomC{$\Sigma'$}
    \doubleLine
    \UnaryInfC{$\phi$}
    \RightLabel{\scriptsize{$PR$}}
    \BinaryInfC{$\chi$}
    \noLine
    \UnaryInfC{$\Sigma'''$}
\end{prooftree}

\bigskip

\begin{prooftree}
    \AxiomC{$\Sigma$}
    \UnaryInfC{$\phi$}
    \AxiomC{$\Sigma'$}
    \doubleLine
    \UnaryInfC{$\phi$}
    \RightLabel{\scriptsize{$PR$}}
    \BinaryInfC{$\top$}
      \RightLabel{\scriptsize{$\top (-)$}}
      \doubleLine
    \UnaryInfC{$\chi$}
    \noLine
    \UnaryInfC{$\Sigma'''$}
    \DisplayProof
    \quad
    $\Longrightarrow$
   \AxiomC{$\Sigma$}
    \UnaryInfC{$\phi$}
    \AxiomC{$\Sigma'$}
    \doubleLine
    \UnaryInfC{$\phi$}
    \RightLabel{\scriptsize{$PR$}}
    \doubleLine
    \BinaryInfC{$\chi$}
    \noLine
    \UnaryInfC{$\Sigma'''$}
\end{prooftree}

\bigskip

In all cases the number of formula occurrences satisfying the original conditions decreases by $1$, which establishes the desired results.

\end{proof}

\begin{proof}[Proof of Theorem~\ref{thm:normalisation}] The proper and permutative reductions are laid out in Figures ~\ref{fig:posproper+reductions}, ~\ref{fig:negproper-reductions} and \ref{fig:permutativereductions}.

\begin{figure}
    \centering

\begin{prooftree}
    \AxiomC{$[\phi]$}
    \noLine
    \UnaryInfC{$\Sigma$}
    \UnaryInfC{$\psi$}             \RightLabel{\scriptsize{$\rightarrow I(+)$}}
    \UnaryInfC{$\phi\rightarrow \psi$}
    \AxiomC{$\Sigma'$}
    \UnaryInfC{$\phi$}    \RightLabel{\scriptsize{$\rightarrow E(+)$}}
    \BinaryInfC{$\psi$} 
       \noLine
    \UnaryInfC{$\Sigma''$}
    \DisplayProof
    $\Longrightarrow$    
    \AxiomC{$\Sigma'$}
    \UnaryInfC{$\phi$}
    \noLine
    \UnaryInfC{$\Sigma$}
    \UnaryInfC{$\psi$}
       \noLine
    \UnaryInfC{$\Sigma''$}
\end{prooftree}
\medskip
\begin{prooftree}
    \AxiomC{$\Sigma$}
    \UnaryInfC{$\phi_1$}
    \AxiomC{$\Sigma'$}
    \doubleLine
    \UnaryInfC{$\phi_2$}    \RightLabel{\scriptsize{$\mapsfrom I(+)$}}
    \BinaryInfC{$\phi_1\mapsfrom \phi_2$}    \RightLabel{\scriptsize{$\mapsfrom E_1(+)$}}
    \UnaryInfC{$\phi_1$}
       \noLine
    \UnaryInfC{$\Sigma''$}
    \DisplayProof
    $\Longrightarrow$     
    \AxiomC{$\Sigma$}
    \UnaryInfC{$\phi_1$}
       \noLine
    \UnaryInfC{$\Sigma''$}
    \DisplayProof
    \qquad
    \AxiomC{$\Sigma$}
    \UnaryInfC{$\phi_1$}
    \AxiomC{$\Sigma'$}
    \doubleLine
    \UnaryInfC{$\phi_2$}    \RightLabel{\scriptsize{$\mapsfrom I(+)$}}
    \BinaryInfC{$\phi_1\mapsfrom A_2$}    \RightLabel{\scriptsize{$\mapsfrom E_2(+)$}}
    \doubleLine
    \UnaryInfC{$\phi_2$}
       \noLine
    \UnaryInfC{$\Sigma''$}
    \DisplayProof
    $\Longrightarrow$     
    \AxiomC{$\Sigma'$}
    \doubleLine
    \UnaryInfC{$\phi_2$}
       \noLine
    \UnaryInfC{$\Sigma''$}
\end{prooftree}
\medskip
\begin{prooftree}
\AxiomC{$\Sigma$}
\UnaryInfC{$\phi_i$}
\RightLabel{\tiny{$I_1\vee(+)$}}
\UnaryInfC{$\phi_1 \vee \phi_2$}
\AxiomC{$[\phi_{1}]$}
\noLine
\UnaryInfC{$\Sigma_1$}
\dottedLine
\UnaryInfC{$\chi$}
\AxiomC{$[\phi_{2}]$}
\noLine
\UnaryInfC{$\Sigma_2$}
\dottedLine
\UnaryInfC{$\chi$}
\RightLabel{\tiny{$E\vee (+)$}}
\dottedLine
\TrinaryInfC{$\chi$}
   \noLine
    \UnaryInfC{$\Sigma'$}
\DisplayProof
$\Longrightarrow$
\AxiomC{$\Sigma$}
\UnaryInfC{$\phi_i$}
\noLine
\UnaryInfC{$\Sigma_i$}
\dottedLine
\UnaryInfC{$\chi$}
   \noLine
    \UnaryInfC{$\Sigma'$}
\end{prooftree}
\medskip

\begin{prooftree}
    \AxiomC{$\Sigma$}
    \UnaryInfC{$\phi$}
    \AxiomC{$\Sigma'$}
    \UnaryInfC{$\psi$}    \RightLabel{\tiny{$ I \wedge (+)  $}}
    \BinaryInfC{$\phi \land \psi$}    \RightLabel{\tiny{$ E_1 \wedge (+)$}}
    \UnaryInfC{$\phi$}
       \noLine
    \UnaryInfC{$\Sigma''$}
    \DisplayProof
    $\Longrightarrow$     
    \AxiomC{$\Sigma$}
    \UnaryInfC{$\phi$}
       \noLine
    \UnaryInfC{$\Sigma''$}
    \DisplayProof
    \qquad
    \AxiomC{$\Sigma$}
    \UnaryInfC{$\phi$}
    \AxiomC{$\Sigma'$}
    \UnaryInfC{$\psi$}    \RightLabel{\tiny{$ I \wedge (+) $}}
    \BinaryInfC{$\phi \land \psi$}    \RightLabel{\tiny{$ E_2 \wedge (+)$}}
    \UnaryInfC{$\psi$}
    \DisplayProof
    $\Longrightarrow$     
    \AxiomC{$\Sigma'$}
    \UnaryInfC{$\psi$}
       \noLine
    \UnaryInfC{$\Sigma''$}
\end{prooftree}

\caption{Proper (proof) reductions}
    \label{fig:posproper+reductions}
\end{figure}

\begin{figure}
    \centering
   \begin{prooftree}
    \AxiomC{$\Sigma$}
    \UnaryInfC{$\phi$}
    \AxiomC{$\Sigma'$}
    \doubleLine
    \UnaryInfC{$\psi$}    \RightLabel{\tiny{$ I \rightarrow (-)$}}
    \doubleLine
    \BinaryInfC{$\phi \rightarrow \psi$}    \RightLabel{\tiny{$ E_1 \rightarrow (-)$}}
    \UnaryInfC{$\phi$}
    \DisplayProof
    $\Longrightarrow$     
    \AxiomC{$\Sigma$}
    \UnaryInfC{$\phi$}
    \DisplayProof
    \qquad
    \AxiomC{$\Sigma$}
    \UnaryInfC{$\phi$}
    \AxiomC{$\Sigma'$}
    \doubleLine
    \UnaryInfC{$\psi$}    \RightLabel{\tiny{$ I \rightarrow (-)$}}
    \doubleLine
    \BinaryInfC{$\phi\rightarrow \psi$}    \RightLabel{\tiny{$ E_2 \rightarrow (-)$}}
    \doubleLine
    \UnaryInfC{$\psi$}
    \DisplayProof
    $\Longrightarrow$     
    \AxiomC{$\Sigma'$}
    \doubleLine
    \UnaryInfC{$\psi$}
\end{prooftree}

\medskip

 \begin{prooftree}
    \AxiomC{$\llbracket \psi \rrbracket$}
    \noLine
    \UnaryInfC{$\Sigma$}
    \doubleLine
    \UnaryInfC{$\phi$}             \RightLabel{\tiny{$ I \mapsfrom (-)$}}
    \doubleLine
    \UnaryInfC{$\phi \mapsfrom \psi$}
    \AxiomC{$\Sigma'$}
    \doubleLine
    \UnaryInfC{$\psi$}    \RightLabel{\tiny{$ E \mapsfrom (-)$}}
    \doubleLine
    \BinaryInfC{$\phi$} 
    \DisplayProof
    $\Longrightarrow$    
    \AxiomC{$\Sigma'$}
    \doubleLine
    \UnaryInfC{$\psi$}
    \noLine
    \UnaryInfC{$\Sigma$}
    \doubleLine
    \UnaryInfC{$\phi$} 
\end{prooftree}

\medskip

\begin{prooftree}
    \AxiomC{$\Sigma$}
    \doubleLine
    \UnaryInfC{$\phi$}
    \AxiomC{$\Sigma'$}
    \doubleLine
    \UnaryInfC{$\psi$}    \RightLabel{\tiny{$ I \vee (-)$}}
    \doubleLine
    \BinaryInfC{$\phi \vee \psi$}    \RightLabel{\tiny{$ E_1 \vee (-)$}}
    \doubleLine
    \UnaryInfC{$\phi$}
    \DisplayProof
    $\Longrightarrow$     
    \AxiomC{$\Sigma$}
    \doubleLine
    \UnaryInfC{$\phi$}
    \DisplayProof
    \qquad
    \AxiomC{$\Sigma$}
    \doubleLine
    \UnaryInfC{$\phi$}
    \AxiomC{$\Sigma'$}
    \doubleLine
    \UnaryInfC{$\psi$}    \RightLabel{\tiny{$ I \vee (-)$}}
    \doubleLine
    \BinaryInfC{$\phi \vee \psi$}    \RightLabel{\tiny{$E_2 \vee (-)$}}
    \doubleLine
    \UnaryInfC{$\psi$}
    \DisplayProof
    $\Longrightarrow$     
    \AxiomC{$\Sigma'$}
    \doubleLine
    \UnaryInfC{$\psi$}
\end{prooftree}

\medskip

\begin{prooftree}
\AxiomC{$\Sigma$}
\doubleLine
\UnaryInfC{$\phi_i$}
\RightLabel{\tiny{$I_1\vee(+)$}}
\doubleLine
\UnaryInfC{$\phi_1 \land \phi_2$}
\AxiomC{$\llbracket \phi_{1} \rrbracket$}
\noLine
\UnaryInfC{$\Sigma_1$}
\dottedLine
\UnaryInfC{$\chi$}
\AxiomC{$\llbracket \phi_{2} \rrbracket$}
\noLine
\UnaryInfC{$\Sigma_2$}
\dottedLine
\UnaryInfC{$\chi$}
\RightLabel{\tiny{$E\vee (+)$}}
\dottedLine
\TrinaryInfC{$\chi$}
\DisplayProof
$\Longrightarrow$
\AxiomC{$\Sigma$}
\doubleLine
\UnaryInfC{$\phi_i$}
\noLine
\UnaryInfC{$\Sigma_i$}
\dottedLine
\UnaryInfC{$\chi$}
\end{prooftree}

\caption{Proper (dual proof) reductions}
    \label{fig:negproper-reductions}

\end{figure}

\begin{figure}
    \centering
    \scriptsize
\hbox{
\begin{bprooftree}
\AxiomC{$\Sigma$}
\UnaryInfC{$\phi \vee \psi$}
\AxiomC{$[\phi]$}
\noLine
\UnaryInfC{$\Sigma'$}
\dottedLine
\UnaryInfC{$\chi$}
\AxiomC{$[\psi]$}
\noLine
\UnaryInfC{$\Sigma''$}
\dottedLine
\UnaryInfC{$\chi$}
\RightLabel{\tiny{$E\vee (+)$}}
\dottedLine
\TrinaryInfC{$\chi$}
\AxiomC{$\Xi$}
\BinaryInfC{$\sigma$}
\DisplayProof
$\Longrightarrow$
\AxiomC{$\Pi_1$}
\UnaryInfC{$\phi \vee \psi$}
\AxiomC{$[\phi]$}
\noLine
\UnaryInfC{$\Sigma'$}
\dottedLine
\UnaryInfC{$\chi$}
\AxiomC{$\Xi$}
\BinaryInfC{$\sigma$}
\AxiomC{$[\psi]$}
\noLine
\UnaryInfC{$\Sigma''$}
\dottedLine
\UnaryInfC{$\chi$}
\AxiomC{$\Xi$}
\BinaryInfC{$\sigma$}
\RightLabel{\tiny{$E\vee (+)$}}
\TrinaryInfC{$\sigma$}
\end{bprooftree}}
\vspace{3em}

\hbox{
\begin{bprooftree}
\AxiomC{$\Sigma$}
\UnaryInfC{$\phi \vee \psi$}
\AxiomC{$[\phi]$}
\noLine
\UnaryInfC{$\Sigma'$}
\dottedLine
\UnaryInfC{$\chi$}
\AxiomC{$[\psi]$}
\noLine
\UnaryInfC{$\Sigma''$}
\dottedLine
\UnaryInfC{$\chi$}
\RightLabel{\tiny{$E\vee (+)$}}
\dottedLine
\TrinaryInfC{$\chi$}
\AxiomC{$\Xi$}
\doubleLine
\BinaryInfC{$\sigma$}
\DisplayProof
$\Longrightarrow$
\AxiomC{$\Pi_1$}
\UnaryInfC{$\phi \vee \psi$}
\AxiomC{$[\phi]$}
\noLine
\UnaryInfC{$\Sigma'$}
\dottedLine
\UnaryInfC{$\chi$}
\AxiomC{$\Xi$}
\doubleLine
\BinaryInfC{$\sigma$}
\AxiomC{$[\psi]$}
\noLine
\UnaryInfC{$\Sigma''$}
\dottedLine
\UnaryInfC{$\chi$}
\AxiomC{$\Xi$}
\doubleLine
\BinaryInfC{$\sigma$}
\RightLabel{\tiny{$E\vee (+)$}}
\doubleLine
\TrinaryInfC{$\sigma$}
\end{bprooftree}}
\vspace{3em}

\hbox{
\begin{bprooftree}
\AxiomC{$\Sigma$}
\doubleLine
\UnaryInfC{$\phi \land \psi$}
\AxiomC{$\llbracket \phi \rrbracket$}
\noLine
\UnaryInfC{$\Sigma'$}
\dottedLine
\UnaryInfC{$\chi$}
\AxiomC{$\llbracket \psi \rrbracket$}
\noLine
\UnaryInfC{$\Sigma''$}
\dottedLine
\UnaryInfC{$\chi$}
\RightLabel{\tiny{$E\land  (-)$}}
\dottedLine
\TrinaryInfC{$\chi$}
\AxiomC{$\Xi$}
\BinaryInfC{$\sigma$}
\DisplayProof
$\Longrightarrow$
\AxiomC{$\Pi_1$}
\doubleLine
\UnaryInfC{$\phi \land \psi$}
\AxiomC{$\llbracket \phi \rrbracket$}
\noLine
\UnaryInfC{$\Sigma'$}
\dottedLine
\UnaryInfC{$\chi$}
\AxiomC{$\Xi$}
\BinaryInfC{$\sigma$}
\AxiomC{$\llbracket \psi \rrbracket$}
\noLine
\UnaryInfC{$\Sigma''$}
\dottedLine
\UnaryInfC{$\chi$}
\AxiomC{$\Xi$}
\BinaryInfC{$\sigma$}
\RightLabel{\tiny{$E\land (-)$}}
\TrinaryInfC{$\sigma$}
\end{bprooftree}}
\vspace{3em}

\hbox{
\begin{bprooftree}
\AxiomC{$\Sigma$}
\doubleLine
\UnaryInfC{$\phi \land \psi$}
\AxiomC{$\llbracket \phi \rrbracket$}
\noLine
\UnaryInfC{$\Sigma'$}
\dottedLine
\UnaryInfC{$\chi$}
\AxiomC{$\llbracket \psi \rrbracket$}
\noLine
\UnaryInfC{$\Sigma''$}
\dottedLine
\UnaryInfC{$\chi$}
\RightLabel{\tiny{$E\land  (-)$}}
\dottedLine
\TrinaryInfC{$\chi$}
\AxiomC{$\Xi$}
\doubleLine
\BinaryInfC{$\sigma$}
\DisplayProof
$\Longrightarrow$
\AxiomC{$\Pi_1$}
\doubleLine
\UnaryInfC{$\phi \land \psi$}
\AxiomC{$\llbracket \phi \rrbracket$}
\noLine
\UnaryInfC{$\Sigma'$}
\dottedLine
\UnaryInfC{$\chi$}
\AxiomC{$\Xi$}
\doubleLine
\BinaryInfC{$\sigma$}
\AxiomC{$\llbracket \psi \rrbracket$}
\noLine
\UnaryInfC{$\Sigma''$}
\dottedLine
\UnaryInfC{$\chi$}
\AxiomC{$\Xi$}
\doubleLine
\BinaryInfC{$\sigma$}
\RightLabel{\tiny{$E\land (-)$}}
\doubleLine
\TrinaryInfC{$\sigma$}
\end{bprooftree}}
\vspace{3em}
    \caption{Permutative reductions. }
    \label{fig:permutativereductions}
\end{figure}

Let $\Pi$ be any deduction which is not in simplified normal form. Transform it in a deduction $\Pi'$ which is in atomic form using Lemma~\ref{atomicpremises} (notice that this procedure may increase the inductive value by copying maximal segments). Now we show that $\Pi'$ can be reduced to a deduction $\Pi''$ in normal form:

Let $\Pi'$ have induction value $\langle d, l \rangle$. Pick any maximal segment $\alpha$ with degree $d$ such that no other maximal segment with degree $d$ has all of its formula occurrences appearing above the last formula occurrence in $\alpha$ or above any formula occurrences side connected with it. Notice that, as usual, such a segment can always be found; if we pick a random segment $\alpha$ with degree $d$ and there is a maximal segment $\beta$ with all formula occurrences above the last formula occurence of $\alpha$ or above a formula occurrence side connected with it, we can just switch to the maximal segment $\beta$, check again if it satisfies the desired conditions and jump to another maximal segment above it if it does not. Since deductions are finite we will always end up with a maximal segment satisfying the desired conditions. Notice also that two segments of the same degree may occur side by side, in the sense that both of them may pass through distinct minor premises of the same rule application, but in this case we are allowed to pick either one.

If the chosen maximal segment has length greater than $1$, we apply a permutative reduction to the application of $\lor (+)$ or $\land (-)$ which concludes its last formula occurrence. Since there are no maximal segments with degree $d$ above formula occurrence side connected with it the inductive value of the deduction is not increased by possible duplications of the deductions of those side connected occurrences. This procedure reduces the length of the segment by at least $1$ (if there are two maximal segments passing through different minor premises of the rule application then both are shortened and the sum of the lengths decreases by $2$, but if there is only one then it decreases by $1$), so we obtain a deduction $\Pi''$ with inductive value $\langle d, s' \rangle$ such that $d = d$ and $s' < s$, hence $\langle d, s \rangle < \langle d, s' \rangle$. 

If the maximal segment has length $1$, we apply a proper reduction to it and obtain a deduction $\Pi''$ with inductive value $\langle d', s' \rangle$. Since there are no maximal segments of the same degree above the last formula occurrences of the segment or side connected occurrences (hence by multiplying occurrences of parts of the deduction we do not increase the inductive value) and any new maximal segments created by the procedure always have lesser degree than that of the original maximal segment, if that maximal segment was the last one with degree $d$ in the derivation then $d' < d$; if it wasn't, the reduction removes a segment of degree $d$ and thus $d = d'$ and $s' = (s - 1)$, so in any case we have $\langle d', s' \rangle < \langle d, s \rangle$. Notice also that neither permutative nor proper reductions add new formulas to the set of premises the deductions depend on. By successively applying this procedure we eventually reach the desired deduction with degree $\langle0, 0 \rangle$. 

Notice that it follows immediately from the shape of the reductions that no new applications of $PR$, $\bot (+)$ or $\top (-)$ with non-atomic conclusions are created, so since $\Pi'$ is in atomic form the deduction $\Pi''$ is in normal form. This allows us to apply Lemma~\ref{MariaLemma}, which also does not create such applications or any new maximal formulas, to obtain a deduction $\Pi'''$ in simplified normal form.

\end{proof}

\begin{proof}[Proof of Corollary~\ref{cor:sub}]
This can be shown by an %straightforward 
adaptation of the proof in~\cite[pp.~52--53]{prawitz1965}. Apart from the addition of dual, analogous cases for each rule, the main difference arises in the adaptation of part (ii) of Prawitz's Theorem~2, where we must also allow formula occurrences that are premises of applications of $PR$ or $\top(-)$, in addition to premises of $\bot(+)$. A further difference is on proofs of co-implication and refutations of implication, which can be treated in a similar way as proofs of conjunction. 
%As the argument contains no substantial new ideas, we omit the details.\qed
%From our proof of normalisation it follows that $\BPR$ indeed satisfies the subformula property. 
\end{proof}

\subsection{Proofs from Section~\ref{sec:bes}}\label{app:sec4}

\begin{proof}[Proof of Theorem~\ref{theorem:positivenegativesupport}]
    We prove both results simultaneously, by induction on the degree of $\phi$. It is not possible to prove both clauses independently due to the interaction between proofs and refutations in some of the semantic clauses.
    \bigskip

    \begin{enumerate}
        \item (Base case). Let $p \in \At$ such that $\Vdash^+_{\B} p$. Note that, in this case, $p$ can never be $\bot$, as this would imply $\vdash^+_{\B}\bot$, contradicting our supposition of logical consistency. Assume, in order to reach a contradiction, that $\Vdash^-_{\B} p$ holds. If $p$ is $\top$, this would entail $\Vdash^-_{\B}\top$, and so $\vdash^-_{\B}\top$, which, once again, contradicts our definition of base. For any other atom $p$, the two conditions yield that $\vdash^+_{\B}p$ and $\vdash^-_{\B}p$ hold simultaneously, so by using one of the rules of epistemic consistency we obtain $\vdash^+_{\B}\bot$ (or $\vdash^{-}_{\mathcal{B}} \top$). Contradiction. Hence, $\nVdash^{-}_{\mathcal{B}} p$. The proof that $\Vdash^{-}_{\mathcal{B}} p$ implies $\nVdash^{+}_{\mathcal{B}} p$ is analogous.
        \medskip 
        \item $(\phi =\psi \wedge \chi).$ Assume that $\Vdash^+_{\B} \psi \wedge \chi$. Then, by the semantic clause for proofs of conjunction, we have $\Vdash^+_{\B} \psi$ and $\Vdash^+_{\B} \chi$. Let $\C$ be any base such that $\C\supseteq\B$ and $\mathcal{D}$ any base such that $\mathcal{D} \supseteq \mathcal{C}$. By transitivity of $\supseteq$ and monotonicity we have $\Vdash^+_{\mathcal{D}} \phi$ and $\Vdash^+_{\mathcal{D}} \chi$, so, by the induction hypothesis, $\not\Vdash^-_{\mathcal{D}} \psi$ and $\not\Vdash^-_{\mathcal{D}} \chi$. Furthermore, since $\mathcal{D}$ is an arbitrary extension of $\mathcal{C}$ we conclude that  $\emptyset; \phi \Vdash_{\C}^-\top$ and $\emptyset; \chi \Vdash_{\C}^-\top$ hold vacuously but, by logical consistency and the semantic clause for refutations of atoms, $\Vdash_{\C}^-\top$ does not hold. Hence, since $\top \in \At$, by the clause for refutation of conjunctions, $\not\Vdash_{\B}^- \phi\wedge \chi$.
        \medskip 
        
        Now, assume that $\Vdash^-_{\B} \psi \wedge \chi$. Then, by the semantic clause for refutation of conjunctions,  for all $\mathcal{C} \supseteq \B$ we have that if $\emptyset ; \psi \Vdash^{-}_{\mathcal{C}} p$ and $\emptyset ; \chi \Vdash_{\mathcal{C}}^{-} p$ then $ \Vdash_{\mathcal{C}}^{-} p$, for all $p \in \At$. Now assume $\Vdash^{+}_{\mathcal{B}} \psi \land \chi$. Then by the clause for proofs of conjunction we have $\Vdash^{+}_{\mathcal{B}} \psi$ and $\Vdash^{+}_{\mathcal{B}} \chi$, whence by monotonicity $\Vdash^{+}_{\mathcal{C}} \psi$ and $\Vdash^{+}_{\mathcal{C}} \chi$ for all $\C \supseteq \B$. Induction hypothesis: $\nVdash^{-}_{\mathcal{C}} \psi$ and $\nVdash^{-}_{\mathcal{C}} \chi$ for all $\C \supseteq \B$. We conclude that $\emptyset ; \psi \Vdash^{-}_{\mathcal{B}} \top$ and $\emptyset ; \chi \Vdash^{-}_{\mathcal{B}} \top$ hold vacuously, so since $\emptyset ; \psi \Vdash^{-}_{\mathcal{C}} \top$ and $\emptyset ; \chi \Vdash_{\mathcal{C}}^{-} \top$ implies $ \Vdash_{\mathcal{C}}^{-} \top$ for all $\C \supseteq \B$ and $\B \supseteq \B$ we have $\Vdash^{-}_{\mathcal{B}} \top$ and so $\vdash^{-}_{\mathcal{B}} \top$. Contradiction. Hence, $\nVdash^{+}_{\B} \psi \land \chi$.
        \medskip
        
        \item $(\phi= \psi \rightarrow \chi).$ Assume that $\Vdash^+_{\B} \psi \rightarrow \chi$. Then, by the semantic clause for proofs of the implication, $\psi;\emptyset\Vdash^+_{\mathcal{B}} \chi$. We take into account the fact that either $\nVdash^{+}_{\B} \psi$ or $\Vdash^{+}_{\B} \psi$. If $\nVdash^{+}_{\mathcal{B}} \psi$ then by the clause for refutation of $\top$ we have $\nVdash^{-}_{\mathcal{B}} \psi \to \chi$. If $\Vdash^{+}_{\mathcal{B}} \psi$ then since $\B \supseteq \B$ and $\psi;\emptyset\Vdash^+_{\mathcal{B}} \chi$ we have $\Vdash^{+}_{\mathcal{B}} \chi$. In this case the induction hypothesis yields $\nVdash^{-}_{\B} \chi$ and then $\nVdash^{-}_{\B} \psi \to \chi$, so in any case we conclude $\Vdash^{-}_{\B} \psi \to \chi$.

       % Then, for every base $\mathcal{C}$ extending $\mathcal{B}$, either $\not\Vdash^+_{\C} \psi$ or, if $\Vdash^+_{\C} B$, $\Vdash^+_{\C} C$. In the first case, we have that $\not\Vdash^+_{\B} B$, in the second case, by the induction hypothesis $\not\Vdash^-_{\C} C$, and so, $\not\Vdash^-_{\B} C$. In either case, we conclude that $\not\Vdash^-_{\B} B\rightarrow C$.
        \medskip

        Now, assume that $\Vdash^-_{\B} \psi\rightarrow \chi$. By the semantic clause for refutations of the implication, $\Vdash^+_{\B} \psi$ and $\Vdash^-_{\B} \chi$. Then, by the induction hypothesis, $\not\Vdash^+_{\B} \chi$. Since $\B\supseteq\B$, $\Vdash^+_{\B} \psi$ and $\not\Vdash^+_{\B} \chi$, we conclude $\psi;\emptyset\not\Vdash^+_{\B} \chi$, hence $\not\Vdash^+_{\B} \psi\rightarrow \chi$. 

    \end{enumerate}
    \medskip
    The proofs for $\vee$ and $\mapsfrom$ are analogous to the ones for $\wedge$ and $\rightarrow$, respectively.

    \end{proof}

\subsection{Proofs from Section~\ref{sec:sc}}\label{app:sec5}

\begin{proof}[Proof of Theorem~\ref{thm:soundness}]
    We prove the result by induction on the length of derivations, understood as the number of rule applications it contains.

    \medskip

\begin{enumerate}
    \item[]  \textbf{Case 1}. The derivation has length $0$. Then it either consists in a single occurrence of proof assumption $\phi$ and is a deduction showing $\phi; \emptyset \vdash^{+}_{\BintNA} \phi$ or in a single occurrence of a refutation assumption $\phi$ and is a deduction showing $\emptyset; \phi \vdash^{-}_{\BintNA} \phi$. It follows that  $\phi; \emptyset \Vdash^{+}_{\mathcal{B}} \phi$ holds for every $\mathcal{B}$, whence $\phi; \emptyset \Vdash^{+} \phi$. It also  follows immediately that  $\emptyset; \phi \Vdash^{-}_{\mathcal{B}} \phi$ holds for every $\mathcal{B}$, whence $\emptyset; \phi \Vdash^{-} \phi$.

    \medskip

     \item[] \textbf{Case 2}. The derivation has length greater than $0$. Then the result is proved by considering the last rule applied in it, so we have to show that each rule application preserves semantic validity. Notice, however, that even though a deduction ending with an application of $\top (+)$ or  $\bot (-)$ has no open assumptions it can still be considered a deduction showing $\Gamma ; \Delta \vdash^{+}_{\BintNA} \top$ or $\Gamma ; \Delta \vdash^{-}_{\BintNA} \bot$, respectively. We only show the case of the rule $PR$ in its proof version, as the dual proof case is analogous.

\end{enumerate}

\begin{itemize}

   \item[] $(PR)$ If $\Gamma_1;\Delta_1\Vdash^+\phi$ and $\Gamma_2;\Delta_2\Vdash^-\phi$, then $\Gamma_1,\Gamma_2;\Delta_1,\Delta_2\Vdash^+ \phi$.
   \medskip

    Assume $\Gamma_1;\Delta_1\Vdash^+\phi$ and $\Gamma_2;\Delta_2\Vdash^-\phi$ and pick an arbitrary $\mathcal{B}$. Assume, for the sake of proving a contradiction, that there is $\mathcal{C} \supseteq \mathcal{B}$ such that $\Vdash^{+}_{\mathcal{C}}\gamma$ and $\Vdash^{-}_{\mathcal{C}}\delta$ for every $\gamma\in\Gamma_1\cup\Gamma_2$ and $\delta\in\Delta_1\cup\Delta_2$. 
     Then, by $\Gamma_1;\Delta_1\Vdash^+\phi$, we have that $\Vdash^{+}_{\mathcal{C}}\phi$, and by $\Gamma_2;\Delta_2\Vdash^-\phi$, we have $\Vdash^{-}_{\mathcal{C}}\phi$, contradicting 
     theorem \ref{theorem:positivenegativesupport}. Hence there is no extension $\mathcal{C}$ of the arbitrary $\mathcal{B}$ such that $\Vdash^{+}_{\mathcal{C}}\gamma$ and $\Vdash^{-}_{\mathcal{C}}\delta$ for every $\gamma\in\Gamma_1\cup\Gamma_2$ and $\delta\in\Delta_1\cup\Delta_2$, so $\Gamma_{1}, \Gamma_{2} ; \Delta_{1}, \Delta_{2} \Vdash^{+} \phi$ holds vacuously.

\end{itemize}
    
    \end{proof}

\begin{comment}
     \Victor{The last step of the proof was incorrect. It was like this: \\ Furthermore, let $\mathcal{C}$ be some base extending $\mathcal{B}$ such that $\Vdash^{+}_{\mathcal{C}}\gamma$ and $\Vdash^{-}_{\mathcal{C}}\delta$ for every $\gamma\in\Gamma_1\cup\Gamma_2$ and $\delta\in\Delta_1\cup\Delta_2$. Then, by $\Gamma_1;\Delta_1\Vdash^+\phi$, we have that $\Vdash^{+}_{\mathcal{C}}\phi$, and by $\Gamma_2;\Delta_2\Vdash^-\phi$, we have $\Vdash^{-}_{\mathcal{C}}\phi$, contradicting Theorem \ref{theorem:positivenegativesupport}. Hence, either $\Gamma_1;\Delta_1\not\Vdash^+\phi$ or $\Gamma_2;\Delta_2\not\Vdash^-\phi$, and so, the statement hold vacuously. \\ Consider the following consequences: $\bot ; \emptyset \Vdash p$ and $ \emptyset ; \top \Vdash p$ and $\bot ; \top \Vdash p$. All of them are vacuously satisfied.}

\end{comment}

\begin{proof}[Proof of Lemma~\ref{lemma:atomicsupportiffderivability}]
    The proof is the same as that for Lemma 7 in \cite{EcumenicalPTS}. The only difference is that there is no need to analyse the cases for $\bot$ and $\top$ separately from the atomic cases, as both of those formulas are now treated as atoms.
\end{proof}

\begin{proof}[Proof of Theorem~\ref{thm:completeness}] Assume that $\Gamma;\Delta\Vdash^* \phi$. Then, $\Gamma;\Delta\Vdash_{\mathcal{U}}^* \phi$. Let $\Theta=(\Gamma\cup\Delta\cup\{\phi\})^{Sub}$, $\alpha$ an atomic mapping for $\Theta$ and $\mathcal{U}$ a simulation base based on $\alpha$. Consider the sets $\Gamma^{\At}=\{\alpha(\psi) | \psi\in\Gamma\}$ and $\Delta^{\At}=\{\alpha(\chi) | \chi\in\Delta\}$. Furthermore, let $\B$ be a base constructed by adding to $\mathcal{U}$ a proof rule concluding $p^{\psi}$ from empty premises, for every $p^{\psi}\in\Gamma^{\At}$, and a refutation rule concluding $p^{\chi}$, for every $p^{\chi}\in\Delta^{\At}$. We consider two cases:

    \begin{enumerate}
        \item $\B$ is inconsistent. Then, (i) there is a deduction $\Pi$ in $\B$ showing $\vdash^+_\B\bot$, or (ii) there is a deduction $\Sigma$ showing $\vdash^-_\B\top$. If (i) is the case but every rule in $\Pi$ is a rule in $\mathcal{U}$, we would have $\vdash^+_{\mathcal{U}}\bot$, contradicting Lemma \ref{lemma:simulationconsistent}. Hence, there must be some rule applied in $\Pi$, which is in $\B$ but not in $\mathcal{U}$. Consider the derivation ${\Pi}'$ we obtain from $\Pi$ by substituting every such rule which is an instance of a proof rule concluding $p^{\psi}$ from empty premises by a pair of assumptions $p^{\psi};\emptyset$ and every instance of a refutation rule concluding $p^{\chi}$ from empty premises by a pair of assumptions $\emptyset;p^{\chi}$. Then, $\Pi'$ is a derivation showing $\Gamma';\Delta'\vdash^+_\mathcal{U}\bot$ for some $\Gamma'\subseteq\Gamma^{\At}$ and $\Delta'\subseteq\Delta^{\At}$. Consider the derivation $\Pi''$ obtained by substituting every instance of a formula $p^{\varphi}$ by $\varphi$ in $\Pi'$. We can easily show, by induction on the length of the derivations, that $\Pi''$ is a derivation showing $\Gamma;\Delta\vdash^+_{\BPR}\bot$. Then, applying $\bot(+)$ at the end of $\Pi''$, we obtain a derivation showing $\Gamma;\Delta\vdash^*_{\BPR}\phi$. If (ii) is the case, we follow a similar strategy and obtain a derivation showing $\Gamma;\Delta\vdash^-_{\BPR}\top$, hence, by applying $\top(-)$ at the end, we show that $\Gamma;\Delta\vdash^*_{\BPR}\phi$.

  \medskip
  
        \item $\B$ is consistent. Then, by construction of this base, we have that $\vdash_\B^+p^{\psi}$ for every  $p^{\psi}\in\Gamma^{\At}$, and $\vdash_\B^-p^{\chi}$ for every  $p^\chi\in\Delta^{\At}$. Note that $\B$ is epistemically adequate, $\Gamma,\Delta\subseteq\Theta$, and so, by Lemma \ref{lemma:atomicsupportiffderivability}, $\Vdash_\B^+{\psi}$ for every  ${\psi}\in\Gamma$, and $\Vdash_\B^-{\chi}$ for every  $\chi\in\Delta$. Furthermore, as $\mathcal{U}\subseteq\B$ and $\Gamma;\Delta\Vdash_{\mathcal{U}}^* \phi$, we get that $\Vdash_{\mathcal{B}}^* \phi$. Again, by Lemma \ref{lemma:atomicsupportiffderivability}, as $\phi\in\Theta$, we obtain $\vdash_{\mathcal{B}}^* p^{\phi}$. Let $\Pi^*$ be a derivation showing $\vdash_{\mathcal{B}}^* p^{\phi}$ in $\B$. If every rule in $\Pi^*$ is a rule in $\U$, then $\Pi^*$ is actually a derivation showing $\vdash_{\U}^* p^{\phi}$. Otherwise, we follow a similar strategy to that of the previous case, substituting every instance of an application of a proof rule concluding $p^\psi$ from empty premisses by a pair of assumptions $p^\psi;\emptyset$, for every $p^\psi\in\Gamma^{\At}$, and analogously for refutation rules concluding $p^\chi$ from empty premises, for $p^\chi\in\Delta^{\At}$. Thus we obtain a derivation in $\U$ showing $\Gamma';\Delta'\vdash_{\U}^* p^{\phi}$ for some $\Gamma'\subseteq\Gamma^{\At}$ and $\Delta'\subseteq\Delta^{\At}$. In either case, we have a derivation showing $\Gamma^{\At}; \Delta^{\At}\vdash_{\U}^* p^{\phi}$. Take this derivation and substitute every instance of an atom $p^{\varphi}$ by $\varphi$. The atoms which are not the image of some formula under $\alpha$, such as (possibly) the conclusion of $\bot(-)$, remain unchanged. It is straightforward to show that the obtained derivation is in $\BPR$ and shows $\Gamma;\Delta\vdash^*_{\BPR} \phi$.
        
\end{enumerate}
\end{proof}

\end{document}